\documentclass[lettersize,journal]{IEEEtran}
\IEEEoverridecommandlockouts
\usepackage[utf8]{inputenc}
\usepackage[english]{babel}
\usepackage{cite}
\usepackage{color,amssymb,amsmath,amsthm,amsfonts,bbm,enumitem}
\definecolor{darkgreen}{rgb}{0.0,0.5,0.0}
\usepackage{algorithmic}
\usepackage[linesnumbered,ruled,vlined]{algorithm2e}
\usepackage{pseudocode}
\SetKwInput{KwInput}{Input}                
\SetKwInput{KwOutput}{Output}              
\usepackage{graphicx}
\usepackage{caption}
\usepackage{stfloats}
\usepackage{url}
\usepackage{verbatim}
\graphicspath{{Figure/}}
\usepackage{multirow}
\usepackage{epstopdf}
\DeclareGraphicsExtensions{.eps}
\usepackage{array}
\usepackage{adjustbox,tabularx, booktabs}
\usepackage{makecell}
\usepackage{footnote}
\makesavenoteenv{tabular}
\newcolumntype{L}[1]{>{\raggedright\let\newline\\\arraybackslash\hspace{0pt}}m{#1}}
\usepackage{authblk}
\newtheorem{theorem}{Theorem}[section]

\theoremstyle{definition}
\newtheorem{definition}{Definition}
\newtheorem{example}{Example}
\theoremstyle{remark}

\usepackage{textcomp}
\usepackage{xcolor}
\usepackage{balance}
\usepackage{hyperref}
\def\BibTeX{{\rm B\kern-.05em{\sc i\kern-.025em b}\kern-.08em
T\kern-.1667em\lower.7ex\hbox{E}\kern-.125emX}}

\usepackage{xr}
\makeatletter

\newcommand*{\addFileDependency}[1]{
\typeout{(#1)}
\@addtofilelist{#1}
\IfFileExists{#1}{}{\typeout{No file #1.}}
}\makeatother

\begin{document}
\title{A Repeated Auction Model for Load-Aware Dynamic Resource Allocation in Multi-Access Edge Computing  \\
\author[]{Ummy Habiba, Setareh Maghsudi, and Ekram Hossain,~\textit{IEEE Fellow}}
\thanks{U. Habiba and E. Hossain are with the Department of Electrical and Computer Engineering, University of Manitoba, Canada (habibau@myumanitoba.ca, ekram.hossain@umanitoba.ca). S. Maghsudi is with the Department of Computer Science, University of Tübingen, and Fraunhofer Heinrich-Hertz-Institute, Berlin, Germany (setareh.maghsudi@uni-tuebingen.de).}}
\maketitle
\begin{abstract}
Multi-access edge computing (MEC) is one of the enabling technologies for high-performance computing at the edge of the 6G networks, supporting high data rates and ultra-low service latency. Although MEC is a remedy to meet the growing demand for computation-intensive applications, the scarcity of resources at the MEC servers degrades its performance. Hence, effective resource management is essential; nevertheless, state-of-the-art research lacks efficient economic models to support the exponential growth of the MEC-enabled applications market. We focus on designing a MEC offloading service market based on a repeated auction model with multiple resource sellers (e.g., network operators and service providers) that compete to sell their computing resources to the offloading users. We design a computationally-efficient modified Generalized Second Price (GSP)-based algorithm that decides on pricing and resource allocation by considering the dynamic offloading requests arrival and the servers' computational workloads. Besides, we propose adaptive best-response bidding strategies for the resource sellers, satisfying the symmetric Nash equilibrium (SNE) and individual rationality properties. Finally, via intensive numerical results\footnote{Source code available at: \url{github.com/titthi/Repeated_GSP_MEC_Offloading}}, we show the effectiveness of our proposed resource allocation mechanism.
\end{abstract}
\section{Introduction}
\label{sec:introduction}
The large-scale deployment of intelligent wireless applications involves numerous computationally-intensive and latency-critical tasks. Although advanced smart devices possess a significant processing capacity, they suffer from limited battery life. Besides, the centralized cloud computing infrastructure is not viable due to long propagation delays, thus a low quality-of-service (QoS). To face this challenge, Mobile Edge Computing, also called Multi-access Edge Computing (MEC) emerges offering computing capabilities at the edge of the radio access network (RAN) \cite{W_Shi_2016}. In the MEC framework, the computing servers are in the proximity of data sources, i.e. user equipment (UE), and computational tasks are executed by harvesting idle computing resources and storage from the edge servers. Thus, the end-to-end communication delay reduces significantly, which makes MEC the prominent choice for provisioning emerging wireless applications. These include over-the-top multi-media streaming services \cite{Barakabitze_2020}, online interactive games \cite{Bilal_2017}, augmented and virtual reality, and tactile internet \cite{Sukhmani_2019}, and video analytic \cite{J_Wang_2018}.

The MEC concept enables wireless service- and infrastructure providers to access heterogeneous fixed and mobile wireless access technologies in WiMax, 4G/LTE, 5G networks, and beyond through software-based services, security solutions, and network functionalities. That facilitates the integration of MEC with the existing 3GPP network architecture without making significant changes to the hardware infrastructure \cite{ETSI_virtualization_2019}. A MEC system is deployable using the RAN elements such as base stations (BSs), access points (APs), and gateways, which typically host the MEC application programming interfaces (APIs) \cite{ETSI_MEC5G_2020}. A MEC system is also implementable in central locations such as data centers of network operators or on moving nodes like passenger vehicles or UAVs. As such, a MEC system can utilize local radio-network contextual information to guarantee secure, reliable, and privacy-preserving services based on intelligent analysis and data processing capabilities at the edge \cite{Amin_2021}.

The smooth operation of the MEC system depends on the availability of computing resources (i.e., the physical or virtual component of CPUs, memory, and storage), given that the task arrives at the MEC servers in random order. Thus, an efficient resource provisioning scheme is required to address the disparity in computing resources at MEC servers with different computational capabilities. On top of that, users' resource demands are dynamic as they vary over time concerning the lengths of computational tasks. The allocation of computing resources involves the partitioning of resource units meeting various resource demands, assigning these resource slices to compute tasks, and finally releasing the allocated resource components upon completion of task execution. Hence the overall resource allocation includes interdependent decision-making processes in a wireless environment, where channel state information varies with time. Therefore, it is crucial to develop an optimal and effective resource provisioning mechanism, with capabilities to make resource allocation decisions dynamically that satisfies varying users' resource demands and computational capabilities of MEC servers \cite{Chen_2018}.

In the context of such a MEC-integrated wireless network, we focus on modeling a MEC service provisioning market using software-defined network (SDN) functionalities. This market incorporates the flexibility to virtualize MEC network infrastructure and resources so a single MEC system can interact with multiple network operators. That enables a multi-vendor multi-access edge computing environment where each MEC service provider (i.e., a vendor) can deploy its application through an individual MEC system and provide services to users of different network operators. It creates new business opportunities in the IT and telecom industry, where rivalry arises among various mobile network operators (MNOs) and MEC service providers. Although the MEC market is currently the fastest-growing software segment \cite{sabella_2021}, policies for network resources-, infrastructure-, and revenue sharing are yet to be standardized for the competitive multi-vendor MEC service trading in wireless networks, i.e., 5G  and beyond.

In this paper, we address the challenges in dynamic resource allocation to provision computation offload as a service under a multi-user multi-vendor MEC business model. In this MEC platform,  several UEs offload computing tasks with heterogeneous QoS requirements, and a number computation service providers participate looking for opportunities to provide offloading as 
a service by executing the offloaded tasks at their computing servers with various computational capabilities. To handle the uncertainty and changing conditions in provisioning the computing resources to offloading UEs, we study the dynamic resource allocation mechanisms which can derive optimal matching decisions between the offloading tasks and computing servers, as well as resource allocation price decisions facilitating the offloading service trading between profit-driven MEC servers and self-interested UEs \cite{Mao_2017}. 

There exist a variety of resource allocation approaches in the literature, mostly focusing on optimizing the resource allocation efficiency using distributed algorithms, combinatorial optimization, game theory, and network utility-based solutions. Among these approaches, auction game theory has been widely applied in MEC offloading services due to its economic properties which establishes linkage between allocation efficiency and fairness in service pricing. However, several research challenges in MEC resource provisioning still remain unaddressed, particularly in modeling the changing conditioning in determining the offloading decision variables during the auction process. Existing online auction-based resource allocation algorithms facilitate meeting UEs' QoEs, but cannot ensure guaranteed profits to MEC service providers. In the presence of multiple service providers, the strategic behavior among the competitive providers has not been addressed while deriving the bidding strategies in the current auction models. Therefore, it is crucial to guarantee that the auction results in equilibrium allocation pricing strategies such that participating sellers and buyers (i.e., MEC servers and UEs) are well off and no bidder (i.e., MEC servers) manipulates the outcomes by misreporting bid prices vindictively. To summarize, the overall MEC service provisioning auction process should be computationally efficient and implementable in a real environment with polynomial time complexity. Additionally, the MEC offloading auction needs to satisfy economic properties, such as allocation efficiency, individual rationality, and incentive compatibility.
\subsection{Contributions}
The key contributions of this work are as follows: 
\begin{itemize}
    \item We propose a service-oriented multi-user multi-vendor MEC system architecture, that deploys computing offloading as a service trading between wireless UEs and MEC servers. 
    \item We present computational workload-aware resource management policies, which implement the offloading task assignment and computing resource allocation process dynamically through an auction. We design the auction model based on a repeated Generalized Second Price (GSP) mechanism, which results computationally efficient, individually rational, and symmetric Nash Equilibrium (SNE) resource allocation solution for every auction round.
    \item In addition, we propose best-response restricted balanced bidding (RBB) strategies for the resource sellers (i.e., MEC servers). We also analyze their adaptive bidding behavior to achieve the SNE equilibrium resource allocations under a dynamic environment.
    \item Besides theoretical analysis, we perform intensive numerical experiments to evaluate the proposed auction mechanism and the quality computation offloading service for different MEC services. The results demonstrate  performance superiority of our proposed RBB bidding strategies and GSP-based resource allocation algorithm compared to the existing welfare-maximizing VCG mechanism.
\end{itemize}
The rest of this paper is organized as follows: the state-of-the-art research is discussed in \textbf{Section~\ref{sec:related_work}}. In \textbf{Section~\ref{sec:system_model}}, we present the multi-user multi-vendor MEC offloading system architecture. We model our computational workload-aware approach to address dynamic task offloading and resource allocation decision problems. In \textbf{Section~\ref{sec:repeated_GSP_mechanism}}, we propose a repeated GSP-based MEC offloading auction mechanism and present the resource allocation and pricing algorithm that determines the offloading and resource allocation decisions. \textbf{Section~\ref{sec:bidding_analysis}} includes the adaptive bidding strategies for the MEC server that satisfy the equilibrium allocations in the dynamic setting, along with individual rationality and computational efficiency. Numerical results are presented in \textbf{Section~\ref{sec:results}}. \textbf{Section~\ref{sec:conclusion}} concludes the paper by adding some remarks. 
\section{Related Work}
\label{sec:related_work}
\begin{table*}[ht!]
\centering
\caption{Overview of studies on auction-based resource allocation mechanism for MEC offloading service provisioning}
\label{tab:MEC_offloading_auction}
\def\arraystretch{1.2}%
\setlength{\arrayrulewidth}{0.5 pt}
\adjustbox{width=\textwidth}{
\begin{tabular}{| p{0.9 cm} | p{2.8 cm} | p{0.6 cm} |p{1.4 cm} | p{1.8 cm} | p{1 cm} | p{1.65 cm} | p{3 cm} | p{0.2 cm} | p{0.2 cm} | p{0.25 cm} |}
    \hline
    \multirow{2}{*}{\textbf{Ref.}} & \multirow{2}{*}{\textbf{Objective}} & \multirow{2}{*}{\textbf{\thead{Users' \\ QoS\textsuperscript{1}}}} & \multirow{2}{*}{\textbf{Resources\textsuperscript{2}}} & \multirow{2}{*}{\textbf{\thead{Auction \\ Model}}} & \multirow{2}{*}{\textbf{\thead{Payment  \\ Rule}}} & \multirow{2}{*}{\textbf{\thead{Matching \\ Model\textsuperscript{3}}}} & \multirow{2}{*} {\textbf{\thead{Allocation  \\ Algorithm}}} & \multicolumn{3}{l|}{\textbf{\thead{Economic \\ Properties\textsuperscript{4}} }} \\ \cline{9-11}
     &    &  &  &  &  &  &  & \textbf{IC} & \textbf{IR} & \textbf{CE} \\
     \hline
     \cite{Sun_2018} & Maximize no. of matched pairs & No & CR & Double auction & Critical Payment & many-to-one & Break-even and dynamic pricing-based heuristic & \checkmark & \checkmark & \checkmark \\
     \hline
     \cite{Q_Wang_2019} & Maximize sum profit of MEC clouds & Yes & CR & Multi-round auction& Vickrey auction & many-to-one & Bid performance ratio-based heuristic & - & - & - \\
     \hline
     \cite{U_Habiba_2019}  & Maximize sum profit of MEC servers & Yes & CR & Online position auction & GSP auction & many-to-one & GSP-based heuristic & - & \checkmark & \checkmark \\
     \hline
     \cite{Q_Li_2020}  & Sum utility of servers and UEs & No & CR & Double auction & winning bid & one-to-one & Experience-weighted attraction based heuristic  & - & - & - \\
     \hline
     \cite{T_Le_2020} & Maximize sum utility of MEC clouds and MNO & No & Bandwidth & Randomized auction & Fractional VCG & many-to-one & Greedy heuristic & \checkmark & \checkmark & \checkmark \\
     \hline
     \cite{Y_Li_2020}  & Maximize long-term social welfare & Yes & $N$-types CR & Online double auction & Critical payment & many-to-one & Matching probability-based heuristic & \checkmark & \checkmark & \checkmark \\
     \hline
     \cite{Z_Shen_2020}  & Maximize sum utility of UEs and SP & No & $N$-types CR & Forward auction & VCG auction & one-to-many & Dynamic programming-based heuristic & \checkmark & - & - \\
     \hline
     \cite{G_Gao_2021}   & Maximize social welfare & Yes & VM & Forward auction& Critical payment & many-to-one & Heuristic with $(\gamma+1)$ approx. ratio & \checkmark & \checkmark & \checkmark \\
     \hline
     \cite{Q_Wang_2022} & Maximize sum profit of MEC clouds & Yes & CR & Online multi-round auction & Vickrey auction & many-to-one & Bid performance ratio-based heuristic  & \checkmark & \checkmark & \checkmark \\
     \hline
     \cite{F_Li_2022}  & Maximize expected utility of macro BS & No &  CR \& WR & second price auction & second price auction & one-to-one & optimal with symmetric Bayesian Nash equilibrium & \checkmark & - & - \\
     \hline
     \cite{L_Zhang_2020_2}   & Maximize sum valuation of QoS & Yes & CR & multi-round auction & VCG auction & many-to-one & Kuhn-Munkras algorithm for bipartite graph & \checkmark & \checkmark & \checkmark \\
     \hline
     \cite{H_Hong_2020} & Maximize social welfare & No & CR & double auction & winning bid & many-to-many & heuristic based on minimum cost flow model & \checkmark & \checkmark & \checkmark \\
     \hline
     \cite{R_Zeng_2020} & Maximize profits of MEC nodes & No & $N$-types CR & reverse auction & first price & many-to-one & heuristic based on expected utility theory & \checkmark & \checkmark & \checkmark \\
     \hline
     \cite{H_Lee_2021} & Maximize profit of MEC SP & No & CR & Myerson auction & virtual bid payment & one-to-one & second price  & \checkmark & \checkmark & \checkmark \\
     \hline
     \cite{Y_Su_2022} & Maximize social welfare & Yes & CR \& WR & Combinatorial auction & Critical payment & many-to-one & Combination of optimal and heuristic & \checkmark & \checkmark & \checkmark \\
     \hline
     Proposed & Maximize sum valuation of resource allocation & Yes & CR & Online position auction & GSP auction & many-to-one & repeated auction-based heuristic & - & \checkmark & \checkmark \\
     \hline
     
\end{tabular}
}
\begin{minipage}{\textwidth}
\vspace{0.05 cm}
\footnotesize{\textsuperscript{1} Users' QoS requirements are considered in terms of service latency and energy consumption constraints.} \\
\footnotesize{\textsuperscript{2} Resources that sellers provide include computing resources (CR), virtual machine (VM) instances, wireless channel resources (WR), and bandwidth.} \\
\footnotesize{\textsuperscript{2} Matching model represents the maximum number of UE (i.e. buyer) that is assigned to a single MEC server/cloud (i.e. seller): \textit{one-to-one} model means one UE can be served at one MEC server/cloud, \textit{one-to-many} model means one UE can be served by more than one MEC servers/clouds at the same time, and \textit{many-to-one} model means multiple UEs can be served at the same MEC server/cloud simultaneously.} \\
\footnotesize{\textsuperscript{3} Essential economic properties in the design of auction mechanism include incentive compatibility (IC), individual rationality (IR), and computational efficiency (CE).}
\end{minipage}
\end{table*}

Recently, MEC has gained popularity as a means to implement a wide range of cloud-based applications and services in wireless networks. It is crucial to utilize the limited resources of the MEC servers efficiently. Existing researches primarily focus on enhancing users' QoE in terms of reduced task execution latency, energy consumption, and offloading cost.  With the growing user demands, the competition among the MEC service providers is also increasing. Majority of the existing works study the MEC offloading system, assuming the MEC servers to be managed a single vendor (i.e. MNO or service provider). Thus, they do not explore the competition among multiple MNOs or service providers. In such a competitive scenario, the vendors may behave strategically to maximize their own profits, by manipulating the allocation prices by misreporting their asking prices. Therefore, it is essential to design an economically efficient resource allocation pricing mechanism, that benefits both the UEs (i.e. resource buyers) and MEC vendors (i.e. resource sellers). \\ \indent
Towards this goal,  the cutting-edge research develops several distributed game- and optimization-based approaches \cite{X_Chen_2016,J_Zhang_2018,H_Guo_2018,M_Liu_2018,Q_Pham_2020}. Among various game-theoretic models, auction is a widely popular approach to study strategic interactions between rational entities in a competitive market scenario. Due to its inherent economic properties and competitive equilibrium strategies, auction model has become the evident choice to develop resource management policies, especially in a dynamic environment in the cloud computing market. In Table~\ref{tab:MEC_offloading_auction}, we summarize the existing auction-based MEC resource allocation mechanisms, comparing the research objectives, QoS criteria, auction framework, and solution approaches. \\ \indent
The state-of-the-art MEC research explores different auction-based resource allocation approaches, addressing the following decision problems: (i) Task-server association, (ii) Resource allocation/provisioning for task-server pairs, and (iii) Computational resource pricing. In the auction framework, the task assignment- and resource allocation problems are often formulated as the winner determination problem (WDP). The problem is then solved using methods from convex optimization, mixed-integer programming, dynamic programming, bipartite graph matching, generalized assignment, etc. For economic efficiency, existing research relies on classic pricing mechanisms such as first-price, second-price, Vickrey-Clarke-Grove (VCG), and critical value-based Myopic auction-based pricing rules. However, there are very few works that study the biddding strategies tackling servers' preferences \cite{Sun_2018,Q_Li_2020,Y_Li_2020,H_Hong_2020,R_Zeng_2020}, and the analysis of strategic stability and convergence of resource allocation pricing policies in dynamic auction game \cite{Q_Li_2020,U_Habiba_2019,R_Zeng_2020}. \\ \indent
Existing auction-based works consider several approaches to address the heterogeneity in MEC resources, such as: jointly allocating computing and wireless resources \cite{F_Li_2022, Y_Su_2022}, defining multi-dimensional computing resources \cite{Y_Li_2020, Z_Shen_2020, R_Zeng_2020}, and partitioning computing resources into virtualized instances \cite{G_Gao_2021}. However, it is still unclear how the MEC system faces the service-specific computational capacity requirements in modeling the MEC resources. In order to address the heterogeneity in offloading users' QoS requirements, very few works investigate the resource allocation problem considering  task deadline constraints \cite{Zhou_2018}, and dynamic transmission delays in the wireless environment \cite{G_Gao_2021}. Nonetheless, they do not investigate the dynamics in the processing servers, whereas the computation offloading performance greatly depends on the current computational workloads of allocated VM resources. \\ \indent
Moreover, existing auctions consider homogeneous pricing strategies, that determine the allocation prices based on the UE-server pairs. Such matching criteria restrict each MEC server to at most one task per user during an offloading period. In a real scenario, however, a single UE may offload multiple tasks with different computing resource requirements. For example, an online gaming UE may simultaneously require computation offloading for the gaming application, VR application, and location-aware navigation application. In such a scenario, to meet the UE's QoS requirements, the offloading tasks should be matched with suitable servers to be processed immediately. A  heterogeneous resource valuation model, can implement the flexibility to handle multiple offloading requests from the same UE, by provisioning computing resources at task-specific allocation prices. Addressing the shortcomings in the literature, we focus on developing a generic service-oriented MEC offloading framework that provisions computing resources to users at competitive prices, and satisfies users with heterogeneous QoS requirements.

\begin{table}[t]
\caption{Key notations}
\label{tab:notation}
\def\arraystretch{1.2}
\centering
\begin{tabular}{|p{1 cm}|p{6 cm}|}
    \hline 
    \textbf{Symbol} &   \textbf{Description}  \\
    \hline 
    $N$ & Number of MEC application processors
    \\
    $J$ & Number of offloading UEs 
    \\ 
    $I$ & Number of MEC servers 
    \\ 
    $\gamma_{i',j}$ & Uplink data rate between UE $j$ and MEC node $i'$ 
    \\ [1ex]
    $\mathcal{R}^n$ & Set of VMs at processor $n$
    \\ 
    $R_i^n$ & Number of VMs of type $n$ at server $i$  
    \\ 
    $r_{im}^n$ & $m$-th VM of type $n$ at server $i$
    \\ [1ex]
    $\mathcal{C}_i^n$ & Computing power of each VM of type $n$ at server $i$
    \\ 
    $W_i^n$ & Number of vCPUs in each VM of type $n$ at server $i$
    \\
    $f_{C_i}^n$ & CPU frequency for a VM of type $n$ at server $i$
    \\ 
    $\beta^n$ & Computational capacity requirements of processor $n$ 
    \\ 
    $\tau_{\max}^n$ & Task completion deadline (msec) in processor $n$ 
    \\ 
    $f_{C_{\min}}^n$ & Minimum CPU frequency required by processor $n$
    \\ 
    $\eta_{i,r_m}^n$ & Workload (MB) of $m$-th VM of type $n$ at server $i$ \\ [1ex]
    $\Gamma_{i,r_m}^n$ & Load per capacity of  $m$-th VM of type $n$ at server $i$ \\ [1ex]
    $\phi_{i,r_m}^n$ & Resource utilization rate of $m$-th VM of type $n$ at server $i$ \\ [1ex]
    $\theta_{i,r_m}^n$ & Expected quality score of $m$-th VM of type $n$ at server $i$ \\ [1ex]
    $v_{i,r_m}^n$ & Valuation of the $m$-th VM of type $n$ at server $i$ \\ [1ex]
    $b_{i,r_m}$ & Bid submitted by server $i$ for the $m$-th VM of type $n$ \\ [1ex]
    $y_{i,r_m}^n$ & Ranking score of the $m$-th VM of type $n$ at server $i$ \\ [1ex]
    $\mathcal{K}^n$ & Task queue with $K=J$ positions at processor $n$
    \\ [1ex]
    $k_{j}^n$ & Computing task of type $n$ offloaded by UE $j$ \\ [1ex]
    $d_{k_j}^n$ & Length of task (MB) offloaded by UE $j$ to processor $n$ \\ [1ex]
    $\lambda_{s_j}^n$ & Task priority index of UE $j$ in processor $n$
    \\ [1ex]
    $x_{s,r_{im}}^n$ & Offloading task-VM matching decision variable \\ [1ex]
    $p_s^n$ & Allocation price ($\$$/VM-hour) decision variable \\ [1ex]
    $u_{i,r_m}^n$ & Utility gain of the $m$-th VM of type $n$ at server $i$ \\ [1ex]
    \hline
\end{tabular}
\end{table}
\section{System Model and Assumptions}
\label{sec:system_model}
%
\begin{figure*}[ht]
    \centering
    \includegraphics[scale=0.6]{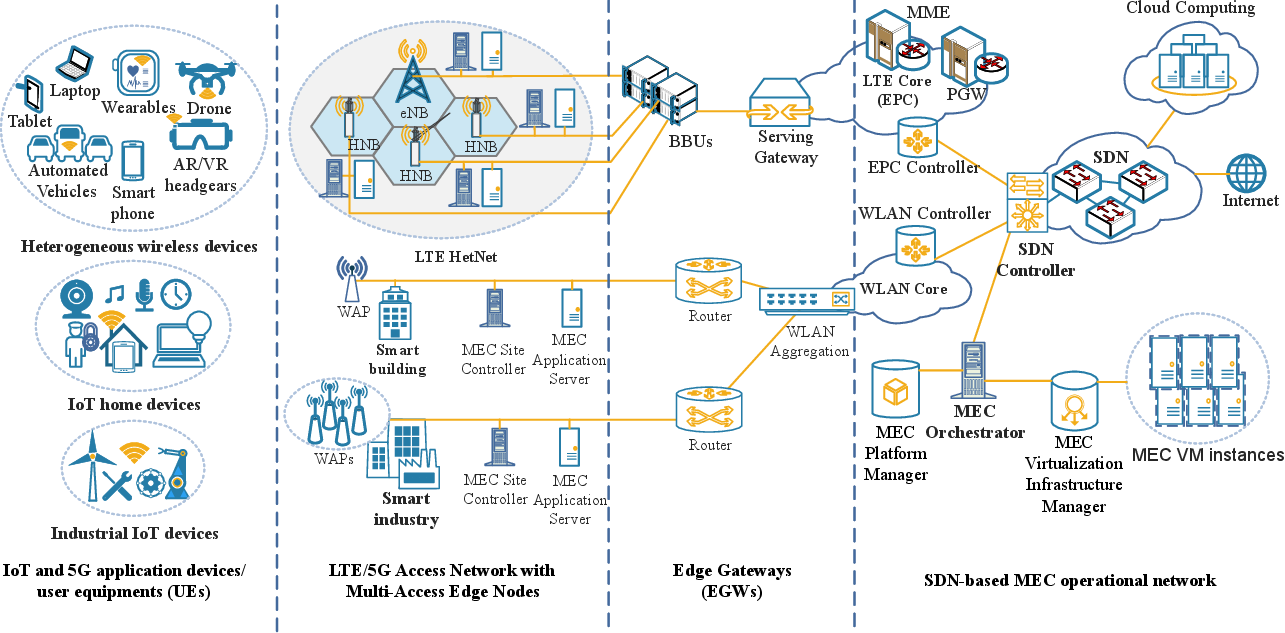}
    \caption{SDN-based multi-tier multi-access edge computing infrastructure and communication network architecture.}
    \label{fig:multi_tier_MEC}
\end{figure*}
We consider a multi-tier multi-access edge computing (MEC) offloading framework as depicted in \textbf{Fig.~\ref{fig:multi_tier_MEC}}. The overall computation and communication functionalities of the network consist of the following tiers:
\begin{itemize}
\item \textbf{Tier 1:} The first tier comprises several heterogeneous 5G wireless application devices within the access network. Each device, referred to as user equipment (UE), is compatible with LTE and WiFi wireless communications standards, such as smartphones.
\item \textbf{Tier 2:} The second tier includes distinct wireless access nodes such as eNodeB (eNB) within the LTE and WLAN access networks. These wireless access points (WAPs) are the MEC nodes/sites deployed onto cellular base stations, buildings, and vehicles. Physical servers and site controllers locate within the premises of the WAPs and gather the computation offloading requests from UEs via the associated WAPs. 
\item \textbf{Tier 3:} In this tier, the computation offloading requests received by the MEC nodes are forwarded to the edge gateways (EGWs) using wireless routers and aggregation switches. The offloading requests are pre-processed at the EGWs and then directed to the MEC system-level controller node in the core network tier. 
\item \textbf{Tier 4:} The core tier is enabled with software-defined network (SDN) capabilities, where a central unit coordinates all the operational functionalities of LTE, WLAN, and MEC via different controller nodes. Due to the scarcity of resources, the servers might be unable to process all of the offloaded tasks. In that case, they forward some of them to the cloud computing platform and central data centers. In our model, a broker or hypervisor, namely, the MEC orchestrator, manages all the MEC components. 
\end{itemize}

A list of the key notations used in this paper in given in Table II.
\subsection{Wireless Network and Communication Model}
In a multi-tier network architecture, we consider $I$ MEC sites gathered in the set $\mathcal{I} = \{i\}_{i=1}^I$, where each site has a WAP, a computing server, and a controller node. We assume that different computation service providers operate each site while competing to earn more revenue by providing computation offloading as services for $N$ applications. The physical servers at the MEC sites have SDN functionalities to host  $N$ application processes simultaneously. We consider a centralized SDN hypervisor, referred to as the \textit{MEC orchestrator} or \textit{orchestrator}. The orchestrator coordinates all the control functionalities across the MEC sites and deploys the computation offloading service between MEC servers and UEs. 

Let $\{j\}_{j=1}^{J}$ be the set of UEs uniformly distributed across the MEC sites, and each UE $j$ is associated with the nearest WAP $i' \in \mathcal{I}$. The UEs get exclusive OFDMA sub-carriers to transmit on the wireless links without interference. For simplicity, we assume the UEs are stationary so the user association remains fixed. Therefore, the same WAP handles all the offloading requests from a user on the MEC site. However, when the offloading requests are forwarded to the MEC system, they can be processed at a different MEC site depending on the tasks' QoS requirements and the server's computational capabilities.

The offloading data rate (in Mbps) in the uplink between the UE $j$ and the associated WAP $i'$ is given by
\begin{align}
   \gamma_{i',j} = \mathrm{BW}_{i'} \log_2 \left( 1 + \frac{P_j^{\mathrm{up}} h_{i',j}}{\sigma_N^2} \right),
\end{align}
where $\sigma_N^2$ is the noise variance, $\mathrm{BW}_{i'}$ the bandwidth of the channel assigned by WAP $i'$, and $P^{\mathrm{up}}_j$ the transmit power of UE $j$.

We consider a generic transmission loss model \cite{ITU_2021}, assuming both UEs and WAPs are below the rooftop level regardless of their antenna heights. Thus the basic transmission loss (in dB) for short-range outdoor communication is given by 
\begin{equation}
h_{i',j} = 10 \mu_d \log_{10} \left( \mathrm{dist}_{i',j} \right) + \mu_0 + 10 \mu_f \log_{10} \left( f_t\right),
\end{equation}
where $\mathrm{dist}_{i',j}$ (in meter) is the distance between UE $j$ and WAP $i'$, and $f_t$ is the wireless channel's operating frequency. Besides, $\mu_d$ and $\mu_f$ are the coefficients that describe the growth of transmission loss with distance and frequency, respectively. Also, $\mu_0$ is the coefficient associated with the offset value of the basic transmission loss. 
\begin{figure}
    \centering
    \includegraphics[scale=0.82]{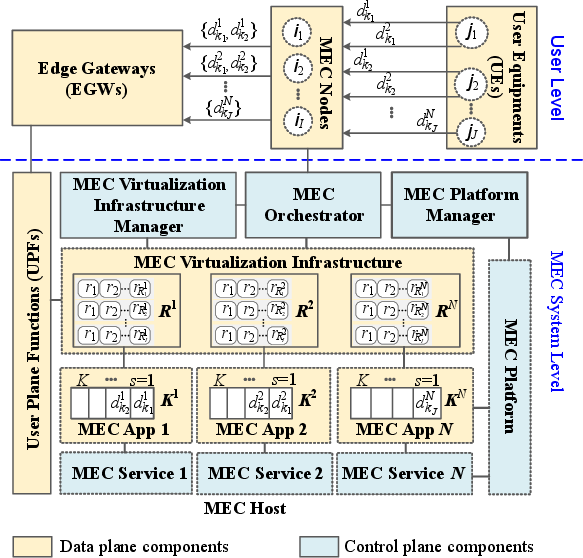}
    \caption{Service-based MEC system architecture.}
    \label{fig:MEC_architecture}
\end{figure}
\subsection{MEC Service Provisioning Model}
We model a service-oriented MEC system architecture \cite{ETSI_MEC_2019} to implement computation offloading as a service. As depicted in \textbf{Fig.~\ref{fig:MEC_architecture}}, the orchestrator is in charge of the overall computation offloading service provisioning process. It starts by receiving offloading requests from the UEs to allocate computing resources to process the offloaded tasks and manage corresponding resource allocation payments. The orchestrator interacts with the MEC sites, EGWs, and UEs through control and user plane functions (UPFs), and implements the offloading services with the help of two other SDN controllers: (a) MEC virtualization infrastructure manager and (b) MEC platform manager.  \\ \indent
The \textit{virtualization manager} mainly oversees the computing infrastructure resources (i.e. servers) of the MEC system. It abstracts and partitions the computational resources into virtualized CPU (vCPU) resource units using SDN functionalities. For each server $i$, the virtualization manager defines $N$ different sets of virtual machine (VM) instances and then allocates vCPUs into these instances aligning with the computational processing requirements for $N$ MEC applications/ services. As shown in\textbf{Fig.~\ref{fig:MEC_architecture}}, each processor $n$ has a resource pool consisting of $R^n = \sum_{i=1}^I r_{i}^n$ VMs, where $R_i^n$ indicates the number of VMs available at server $i$ to process the tasks of application $n$. We represent the $m$-th VM at server $i$ by $r_{im}^n$, which has computing power $C_{i}^n$ (MBps) and $W_{i}^n$ vCPUs each with CPU frequency $f_{C_i}^n$ (in Hz). \\ \indent
The \textit{platform manager} supervises the offloading task execution process by managing the $N$ MEC application processors via the MEC platform. The platform maintains a task queue $\mathcal{K}^n$ for each processor $n$ to handle the incoming offloading requests for application $n$. Each queue $\mathcal{K}^n$ has a fixed $(K=J)$ number of positions. Each position $s$ is associated with a user-specific task priority index $\lambda_{s_j}^n$ that helps prioritize the requests. \\ \indent
To complete each offloading task of application type $n$ within the deadline $\tau_{\max}^n$ (msec), the minimum average processing speed (MBps) for the UE $j$'s request yields $f_{C_{\min}}^n \geq (\Bar{d}_{j}^n / \tau_{\max}^n)$ \cite{A_Alnoman_2021}, where $\Bar{d}_{j}^n$ (MB) represents the average length of tasks offloaded by UE $j$. Therefore, we consider a delay-aware task prioritizing policy defining the task priority index for the $s$-th task position as 
\begin{equation}
\label{eqn:priority}
\lambda_{s_j}^n =  \dfrac{\Bar{d}_{j}^n}{\tau_{\max}^n},
\end{equation}
which prioritizes tasks with shorter task completion deadlines and requested by UEs who offloads tasks with larger sizes on average. Thus, the platform manager filters the incoming offloading requests according to the requested task type $n$, and forwards the requests to the corresponding processor $n$ and places them into the task queue $\mathcal{K}^n$ according to the requesting UE $j$'s priority index  $\lambda_{s_j}^n$. The orchestrator then coordinates the auction process that matches the tasks and suitable VMs. We consider non-preemptive priority-based task assignments at each processor, which means allocated VMs are not released until task processing finishes. 
\subsection{MEC Orchestration Model}
%
\begin{figure*}
    \centering
    \includegraphics[scale=0.9]{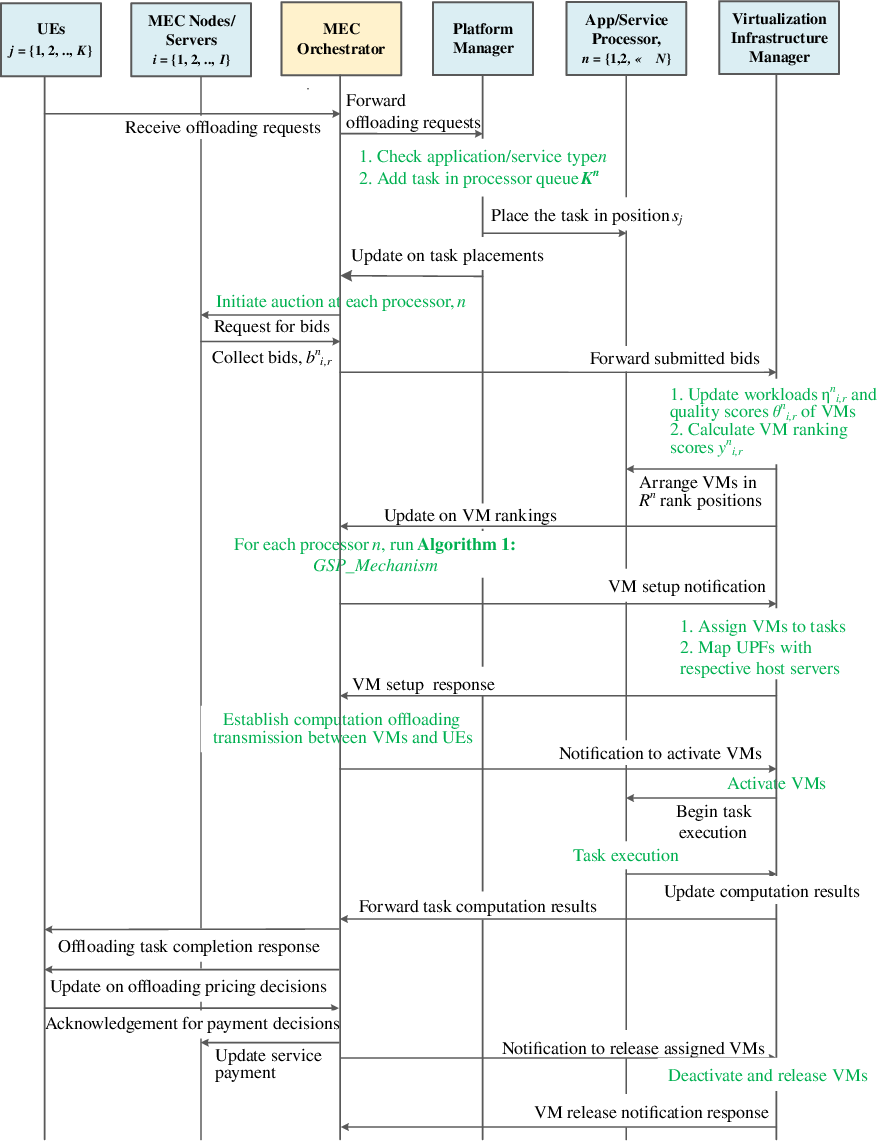}
    \caption{Workflow of the computation offloading mechanism in MEC framework.}
    \label{fig:MEC_workflow}
\end{figure*}
The orchestrator manages the end-to-end computation offloading service provisioning process through an auction. To model the dynamic conditions in computation offloading in MEC, we design the overall orchestration process based on an online or dynamic auction game game between two sets of players: UEs and MEC servers, while the orchestrator is the auctioneer. It runs the auction mechanism at discrete offloading time slots, $t=1, 2, \dots$, with a duration of $\Delta t$, at $N$ processors. Thus, it obtains $N$ sets of independent allocations. For each auction round in processor $n$, $K=J$ UEs are the resource buyers. The $I$ servers are then the sellers with $R^n$ VMs as the auction commodities. \\ \indent
Formally, at processor $n$ and time slot $t>0$, we express the offloading service provisioning state as $\langle \mathcal{K}^n (t), \mathcal{R}^n (t), \mathcal{X}^n (t), \mathcal{P}^n (t) \rangle$, where
\begin{itemize}
    \item $\mathcal{K}^n (t)$ represents the task queue with $K$ distinct positions. When the UE $j$'s request arrives at the processor $n$, its computing task $k_j^n$ with length $d_{k_j}^n$ (MB) takes position $s_j \in \mathcal{K}^n (t)$ with the priority index $\lambda_{s_j}^n$, 
    \item[-] $\mathcal{R}^n (t)$ represents the list of ranked VMs in processor $n$, where each VM $r_{im}^n$ has a position $m$. The position is updated every time slot based on the current computing processing score $\theta_{i,r_{m}}^n (t)$ and bid $b_{i,r_m}^n (t)$, 
    \item[-] $\mathcal{X}^n (t)$ represents the offloading task-VM matching decisions, where $x_{s,r_{im}}^n(t) \in \{0,1\}$ is the decision variable indicating if the task in position $s$ is assigned to the $m$-th VM at server $i$ in time slot $t$, 
    \item[-] $\mathcal{P}^n (t)$ represents the allocation pricing decisions. The decision variable $p_{s}^n(t) \geq 0$ indicates the amount (in $\$$/VM-hour) that the UE pays for offloading the task in position $s \in \mathcal{K}^n(t)$.
\end{itemize}
We describe the orchestration stages below and summarize the overall workflow in \textbf{Fig.~\ref{fig:MEC_workflow}}. 
\begin{itemize}
\setlength\itemsep{0.25em}
    \item[(i)] \textbf{Dynamic Queuing of Incoming Offloading Tasks}: At the beginning of time slot $t$, the orchestrator gathers the incoming offloading requests and forwards them to the platform manager. The platform manager checks their task types and arranges them into respective processing queues $\mathcal{K}^n(t)$. The offloading requests arrive dynamically, i.e. the orchestrator does not have any prior information about the number of tasks in the current time slot until the requests arrive.
    \item[(ii)] \textbf{Dynamic Resource Management Based on Computational Workloads}: The virtualization manager uses dynamic resource management policies and maintains separate queues to monitor the computational workloads of VMs at each processor $n$. At the beginning of every time slot $t$, it checks the status of VMs and updates their workloads in the queue denoted by $\eta_{i,r_m}^n(t)$. The workload of each VM queue thus evolves as
    \begin{align}
    \eta_{i,r_m}^n(t) = \eta_{\mathrm{rem},i,r_m}^n(t-1) + \sum_{s=1}^K d_{k_j}^n (t) x_{s,r_{im}}^n(t), 
    \end{align}
    where 
    {\small
    \begin{align*}
        \eta_{\mathrm{rem},i,r_m}^n(t-1) = \max \left \{ \sum_{s=1}^K d_{k_j}^n (t) x_{s,r_{im}}^n(t-1) - C_i^n \Delta t, 0 \right \} 
    \end{align*}
    } \\
    represents the remaining workload of the VM $r_{im}^n$ after computing the task assigned in round $(t-1)$. Therefore, the current computational workload per capacity (Byte per CPU cycle (BPC)) on the $m$-th VM at server $i$ can be estimated as \cite{Souravlas_2022} \\
    \begin{equation}
        \Gamma_{i,r_m}^n(t) =  \dfrac{\eta_{i,r_m}^n(t)}{C_i^n \Delta t}.
    \end{equation}
    Based on the VM's workload per capacity, the resource utilization metric is given by
    \begin{footnotesize}
    \begin{align}
    \label{eqn:VM_utilization}
        \phi_{i,r_m}^n(t) = \left \{ \begin{array}{l l}
         0, & \hspace{-1em} \text{if } \Gamma^{\max} \leq \Gamma_{i,r_m}^n(t)   \\
         \dfrac{ \Big \vert \Gamma_{i,r_m}^{n}(t) - \Gamma^{\max}   \Big \vert}{\Gamma^{\max}} , & \hspace{-1em} \text{if } \Gamma^{\mathrm{\min}} \leq \Gamma_{i,r_m}^{n}(t) < \Gamma^{\max}    \\
         1, &  \hspace{-1em} \text{otherwise},
        \end{array} \right. 
    \end{align}
    \end{footnotesize}
    \\
    where $\Gamma^{\max}$ indicates the maximum load allowed per capacity on each VM. Task assignment beyond this limit would overload the VM with the lowest utilization score (i.e. $\phi_{i,r_m}^n = 0$). In contrast, it scores the highest (i.e. $\phi_{i,r_m}^n = 1$) when the VM is underloaded, i.e. the existing load is smaller than the minimum resource utilization threshold $\Gamma^{\mathrm{\min}}$; otherwise, the scoring function determines the resource utilization scores for the VMs under normal workload between 0 and 1. \\
    The expected computation performance quality score of a VM is updated in each time slot $t$, according to its current resource utilization as
    \begin{equation}
       \theta_{i,r_m}^n(t) = \dfrac{W_i^n f_{C_i}^n}{f_{C_{\min}}^n} \phi_{i,r_m}^n (t).
    \end{equation}
    \item[(iii)] \textbf{Collecting Bids from Servers}: After queuing the incoming offloading requests and updating VMs' expected computation performance quality score, the orchestrator initiates the auction with the bid collection process by requesting servers to submit bids for each processor $n$. Also, it provides the servers with information on service provisioning state values from the previous round of the auction. \\
    The servers determine their bids for each round, following their best-response bidding strategies (as we will discuss in \textbf{Section~\ref{sec:bidding_analysis}}). For each server $i$, we use $\mathbf{b}_{i}^n (t) = \left[ b_{i,r_1}^n, b_{i,r_2}^n, \dots, b_{i,r_{R_i^n}}^n \right]$ to denote its bids for its VMs in processor $n$ in round $t$. 
    \item[(iv)] \textbf{Determining Resource Allocation Decisions}: After receiving the updated bids, the orchestrator runs the resource allocation and pricing algorithm at each processor $n$. We consider GSP-based resource allocation and pricing rules to determine the task assignment decisions, $x_{s,r_{im}}^n (t)$, and VM allocation pricing decisions, $p_s^n(t)$. We outline the resource allocation and pricing in \textbf{Algorithm~\ref{algo:GSP_allocation_pricing}}.
    \item[(v)] \textbf{Executing Offloading Tasks at Servers}: In the next stage, the orchestrator notifies the virtualization manager to set up the VMs according to the task-VM matching decisions. The virtualization manager then maps the VMs and tasks to the host servers and UPFs, respectively. It updates the orchestrator after setting up the VMs and then activates the VMs to execute tasks after receiving the notification from the orchestrator. In the end, it sends the computed results back to the UEs via the orchestrator. \\
    After sending the task computation results, the orchestrator updates the UEs about the total offloading service payment information and notifies the virtualization manager to release the assigned VMs. Once the UEs acknowledge the service billing, the orchestrator ends the current offloading transmission session. We consider \textit{pay-per-CPU cycle} payment methods for processing offloading tasks, where payments are collected at the end of the billing cycle (e.g. bi-weekly, monthly, annually). 
\end{itemize}
\subsection{Utility Model of MEC Servers}
We consider the utility model for the servers based on the profits earned from VM allocation in each auction round. The profit for each VM is given by the price difference between the allocation price settled by the auctioneer (orchestrator) and the VM's private valuation. \\ \indent
Let $v_{i,r_m}^n$ be the private valuation ($\$$/VM-hour) of the $m$-th VM at server $i$ in processor $n$, and is determined based on the CPU power consumption at the processor \cite{J_Zhang_2018}:
\begin{equation}
\label{eqn:valuation}
   v_{i,r_m}^n = \rho_i \kappa W_{i}^n  \left(f_{C_i}^n \right)^2, 
\end{equation}
where $\rho_i$ is a scaling parameter to convert the CPU power consumption into monetary value. Besides, $\kappa$ is the effective switched capacitance of the processor. 
Therefore, if server $i$ allocates $m$ VMs to task position $s$ in processor $n$ in time slot $t$, its profit is
\begin{equation}
\label{eqn:mec_utility}    
    u_{i,r_m}^{n} (t) =  \sum_{s=1}^{K} \lambda_{s_j}^n \theta_{i,r_s}^n(t) \left( p_{s}^n (t) - v_{i,r_s}^n \right) x_{s,r_{i(m=s)}}^{n}(t).
\end{equation}
The total utility of server $i$ in time slot $t$ thus yields
\begin{equation}
   U_i(t) = \sum_{n=1}^{N} \sum_{m=1}^{R_i^n} u_{i,r_m}^n(t).
\end{equation}
%
\subsection{Utility Model of Offloading Users}
The utility of the offloading users depends on their QoE in terms of both task offloading cost and task execution latency. The total execution time of task $k_j^n$ offloaded by UE $j$ consists of (i) upload, (ii) queue at the MEC platform, (iii) computation at the allocated VM, and (iv) sending back the results. Often, the size of the computed results is negligible compared to uploading the data; Hence, we ignore the time to send back the results in the downlink. Therefore, the end-to-end offloading service latency (in sec) for the task $k_j^n$ in time slot $t$ can be written as
\begin{equation}
    \delta_{k_j}^n (t) = \delta_{k_j,\mathrm{up}}^{n} (t) + \delta_{k_j,\mathrm{wait}}^{n} (t) + \delta_{k_j, \mathrm{comp}}^{n} (t) \, , 
\end{equation}
where for UE $j$, the upload time to MEC node $i'$ is given by $\delta_{k_j,\mathrm{up}}^{n} (t) = \frac{d_{k_j}^n (t)}{\gamma_{i',j}} (t)$. At server $i$, the computation time is $\delta_{k_j, \mathrm{comp}}^{n} (t) = \frac{d_{k_j}^n (t)}{C_i^n}$. The waiting latency for each request depends only on the computation time of the tasks placed ahead in the task queue. Thus, for a task in position $s \in \mathcal{K}^n(t)$, it can be estimated as $ \delta_{k_j,\mathrm{wait}}^{n} (t) = \sum_{s'=1}^{s-1} \delta_{k_{s'_j, \mathrm{comp}}}^{n}(t)$, where $k_{s'_j}$ denotes the task in an upper position $(s'< s)$ in the queue. \\ \indent
To quantify UE $j$'s level of satisfaction with the overall service latency for a task $k_j^n$, we define the following performance metric:
\begin{align}
    \alpha_{k_j}^n (t) = \left \{ \begin{array}{l l}
         \dfrac{\vert \tau_{\max}^{n} - \delta_{k}^n (t) \vert}{\tau_{\max}^{n}}, & \text{if } 0 < \delta_{k_j}^n \leq \tau_{\max}^{n} \\
         0, & \text{otherwise. }
    \end{array} \right. 
\end{align}
%
Therefore, UE $j$'s QoE in terms of offloading service latency can be estimated as the following mean opinion score
\begin{equation}
    Q_j^{\mathrm{latency}} (t) =  \frac{1}{N} \left [  \sum_{n=1}^{N} \sum_{s=1}^{K} \sum_{r=1}^{R^n}  \alpha_{k_j}^n (t) \, x_{s,r_{im}}^n(t), \, \right ].
\end{equation}
%
Next, considering $\Bar{a}_j$ ($\$$) as the UE $j$'s monetary budget, we model the UE's QoE in terms of offloading service cost using the budget cost savings ratio as 
\begin{equation}
    Q_j^{\mathrm{cost}} = \bigg\vert \frac{\Bar{a}_j - a_j(t)}{\Bar{a}_j} \bigg\vert,
\end{equation}
where $a_{j} (t)$ represents the cost (in \$) to process the offloading tasks, which depends on the processing time spent by the allocated VMs. Formally, 
\begin{equation}
    a_{j} (t) =  \sum_{n=1}^N \sum_{s=1}^{K} \sum_{r=1}^{R^n} \delta_{k_j, \mathrm{comp}}^n(t) \, p_{s}^{n}(t) \, x_{s,r_{im}}^n(t).
\end{equation}
The overall utility gain of the UE $j$ thus becomes 
\begin{equation}
    Q_j = q_l \, Q_j^{\mathrm{latency}} + q_c \, Q_j^{\mathrm{cost}}, 
\end{equation}
where $q_l \in [0,1]$ and $q_c \in [0,1]$ are the QoE coefficients to adjust the trade-off between the offloading service latency and offloading cost, respectively.
\section{GSP-Based Auction Mechanism Design for Dynamic Computation Offloading and Resource Allocation}
\label{sec:repeated_GSP_mechanism}
Concerning the MEC orchestration model, we describe the auction mechanism step-by-step, establishing computation offloading service trading between UEs and MEC service providers. We outline the MEC offloading auction as an online auction where the number of commodities (i.e., resources) and bid values may change over different time slots $t$ depending on the incoming offloading requests. In each time slot $t$, the auction is simply a static game with complete information where the orchestrator runs a resource allocation algorithm and determines the offloading task assignment and resource allocation pricing decisions using the current channel and task information.

As we aim to maximize the total valuation of resource allocation in the dynamic setting, we consider the repeated GSP auction mechanism where the static auction game is repeated over different time slots, with varying input information on offloading tasks. That means the same GSP-based resource allocation algorithm is run in every time slot, which gives offloading decision variables based on submitted bids and offloading tasks in the current time slot. Thus, the repeated auction setting establishes a dynamic auction game with incomplete information, allowing bidders (i.e., MEC servers) to optimize their bidding strategy to win more offloading tasks by updating their bids values prior information, e.g., resource allocation outcomes from the previous auction round.
In the following section, we start formulating the resource allocation optimization problem assuming the static version of computation offloading in a single auction round in time slot $t$. We present the winner determination problem (WDP), detailing offloading decision variables and resource allocation constraints, and then address the solution approaches. Later, we explain the GSP-based resource allocation algorithm followed by the dynamic version of resource allocation and pricing mechanism in the repeated auction setting.
\subsection{Winner Determination Problem Formulation}
We consider a static computation offloading scenario at the $n$-the processor. There are $K$ computing tasks in the queue $\mathcal{K}^n$ and $I$ servers, each having a distinct set of VM resources, $\{R_i^n \}$. Given the quality scores of the VMs, $\theta_{i,r_m}^n$, and the bids submitted by the servers, $b_{i,r_m}^n$, the orchestrator formulates the following optimization problem that maximizes the total allocation valuation by determining the winners (i.e. VMs) for each processor $n$.
\begin{align}
\label{eqn:opt_prob}
  \nonumber & \max && \sum_{s=1}^{K} \,  \sum_{i=1}^I \underset{r_{im} \in R_i^n}{\sum} z_{s,r_{im}}^n\, x_{s,r_{im}}^n \\  
  \nonumber &  \mathrm{s.t.} \quad && (\mathrm{C1})~  \sum_{i=1}^I \underset{r_{im} \in R_i^n}{\sum}  x_{s,r_{im}}^{n} = 1, \quad \forall \, s\\
  \nonumber &&& (\mathrm{C2}) ~~ \overset{K}{\underset{s=1}{\sum}} \,  \underset{r_{im} \in R_i^n}{\sum} x_{s,r_{im}}^{n} \leq R_{i}^n, \quad \forall \, i \\
   &&&  (\mathrm{C3})~ \;  x_{s,r_{im}}^{n} \in \{0,1\}, \, \forall \, s, r_{im},
\end{align}
where $z_{s,r_{im}}^n = \lambda_{s_j}^n \theta_{i,r_m}^n/ b_{i,r_{m}}^n$ represents the allocation valuation for the VM $r_{im} \in R_i^n$, who wins the $s$-th offloading task. In (\ref{eqn:opt_prob}), constraint $(\mathrm{C1})$ guarantees that each offloading task is matched with exactly one VM. The constraint $\mathrm{(C2)}$ ensures that the total number of tasks assigned to a server $i$ does not exceed its VM resource constraints. Furthermore, $(\mathrm{C3})$ means that the  offloading task assignments are binary decision variables. 
\subsection{WDP Solution Approaches}
The orchestrator determines the corresponding resource allocation prices, $\{p_s^n\}$, based on the offloading task assignment decisions, $\{x_{s,r_{im}}^n\}$ obtained through solving the WDP in (\ref{eqn:opt_prob}). To achieve socially-efficient allocation outcomes, one can adopt the classic VCG pricing mechanism \cite{Vickrey_1961} so that the allocation algorithm solves the WDP by selecting the VMs that gives maximum allocation valuation for each offloading task. It then settles the allocation prices for each bidder equivalent to the amount it contributes to social welfare. However, the VCG mechanism is often unsuitable for practical auction design, especially for time-sensitive computation offloading services in MEC. The reasons include NP-hardness of WDP, revenue deficiency, and difficulties handling the bidders' information when the auction is part of a larger sequence of commercial transactions \cite{M_Rothkopf_2007}. \\ \indent
To address the computational complexity of the WDP, we design an approximation algorithm that finds the allocation decisions in polynomial time. We notice that (\ref{eqn:opt_prob}) is an instance of a multidimensional multiple-choice knapsack problem (MDMCKP) \cite{Kellerer_2004}. In that problem, a single knapsack consists of $I$ containers/dimensions, each dimension with a resource constraint of $R_i^n$, and the knapsack packs exactly one task from each offloading UE in $s \in \mathcal{K}^n$ applying the "multiple-choice" constraint of MDMCK. It is well-known that MDMCKP is an NP-hard problem \cite{Kelly_2004}. Therefore, assuming a single dimension by relaxing the resource constraints in $\mathrm{(C2)}$, (\ref{eqn:opt_prob}) boils down to an MCKP. Although MCKP is still NP-hard, it is solvable in pseudo-polynomial time using dynamic programming as long as the number of choices is low for each item \cite{Kelly_2004}. \\ \indent
Thus we use the knapsack dynamic programming-based resource allocations as the upper bound to the WDP (\ref{eqn:opt_prob}). Afterward, we apply the VCG pricing rules to obtain a benchmark solution to the resource allocation prices for computation offloading services. However, still, a more practical auction design to support dynamic resource allocations and the long sequences of service-oriented payment transactions for online auctions in MEC. Hence, we develop a computationally efficient and practically viable repeated auction model in the subsequent section, using the features from dynamic position auction and GSP mechanism \cite{Varian_2007}.
%
\subsection{GSP-Based Resource Allocation and Pricing Mechanism }
%
In this section, we first outline a GSP-based allocation mechanism addressing the WDP (\ref{eqn:opt_prob}). We summarize the modified GSP allocation and pricing algorithm in \textbf{Algorithm~\ref{algo:GSP_allocation_pricing}} that determines the offloading task assignment decisions, $x_{s,r_{im}}^n$, and corresponding allocation prices, $p_s^n$, given the tasks' priority scores, VMs' quality scores, and bids, as inputs.  \\ \indent
For each processor $n$, we assume $K$ distinct task positions following the position auction framework. Each position has a task priority index $\lambda_{s_j}^n$. So, the proposed mechanism first arranges the tasks in decreasing order of the requesting UEs' task priority indices as in $\lambda_{s_1}^n \geq \lambda_{s_2}^n \geq \dots \geq \lambda_{s_K}^n$. \\
Besides, we consider the pool of VM resources, $\mathcal{R}^n$, as the list of items to allocate to the offloading task positions. We also define a function to rank the VMs according to their expected computation performance quality scores, $\theta_{i, r_m}^n$, and the bids, $b_{i, r_m}^n$, submitted by their host server $i$. The ranking score of the $m$-th VM at server $i$ is given by
\begin{equation}
\label{eqn:ranking}
   y_{i, r_m}^n  = \frac{\theta_{i, r_m}^n}{b_{i, r_m}^n}.
\end{equation}
Therefore, the proposed mechanism arranges the VMs into distinct rank positions as in: $y_{i, r_1}^n \geq y_{i,r_2}^n \geq \dots \geq y_{i, r_{R^n}}^n$.
\\ \indent
Next, the mechanism sequentially matches the tasks and VMs according to their positions using the GSP auction allocation rules. That is, the task in the $i$-th position is matched to the VM in the $i$-th rank. Besides, according to the GSP pricing rule, the corresponding allocation price is equal to the bid that the winning VM $r_{im}^n$ requires to maintain its current $m$-th ranking position. Hence, the allocation price for the task position $s$, matched to the VM at rank $m=s$, satisfies
\begin{equation}
    b_{i,r_{s}}^n \leq \frac{\theta_{i,r_s}^n}{\theta_{i,r_{s+1}}^n} b_{i,r_{s+1}}^n. 
\end{equation}
Thus, we define the price adjustment rate with respect to the VM in the $m$-th rank position as
\begin{equation}
    \Theta_m^n = \frac{\theta_{i,r_m}^n}{\theta_{i,r_{m+1}}^n}.
\end{equation}
The monotonicity in allocation prices means that a VM that wins a task in a higher position receives more than lower VMs. To satisfy that, we arrange the price adjustment rates in a decreasing order into $\mathbf{\Theta}_{R}^n$. \\ \indent
Finally, the proposed modified GSP mechanism settles the allocation price for the $s$-th position as
\begin{equation}
\label{eqn:GSP_price}
    p_s^{n} = \left\{ \begin{array}{lcl} 
   \Theta_{R_s}^n b_{i,r_{s+1}}^n, && \text{if } 1 \leq s < K \\
   b_{i,r_s}^n + \epsilon, && \text{if } s = R^n
\end{array} \right.
\end{equation}
where $\epsilon$ is a small positive constant to guarantee that a VM is paid more than its bid, when it is the last one in the ranked list (i.e. $R^n=K$).
%
\begin{algorithm}[t]
\caption{$\mathit{GSP\_Mechanism}$}
\label{algo:GSP_allocation_pricing}
\KwInput{$\lambda^{n}, \theta^n, \textbf{\textit{b}}^n$}
\KwOutput{$\pi^n= (\textbf{\textit{x}}^n, \textbf{\textit{p}}^n)$}
Update the tasks' positions in the queue according to their priority indices, 
$$\mathcal{K}^n \gets sort(\lambda^n, \mathrm{descend})$$ \\
Find each VM's ranking score,\\
\For{$m = 1 \; \mathbf{to} \; R^n$}
{ $y_{i,r_m}^n = \dfrac{\theta_{i,r_m}^n }{b_{i,r_m}^n}$
}
Arrange VMs according to their ranking scores, $\mathcal{R}^n \gets sort(\textbf{\textit{y}}^n, \mathrm{descend})$ \\
Find the price adjustment rates,\\
\For{$m \gets 1 \; \mathbf{to} \; (R^n-1)$}
    {
    $r_m \gets \mathcal{R}^n[m]$, $r_{(m+1)} \gets \mathcal{R}^n[m+1]$,\\
    $\Theta_{m}^n = \frac{\theta_{i,r_m}^n }{\theta_{i,r_{m+1}}^n}$ \\
    }
Arrange price adjustment rates in decreasing order,
$ \quad \mathbf{\Theta}_{R}^n \gets sort(\mathbf{\Theta}^n, \mathrm{descend})$ \\
Sequentially match the tasks in $\mathcal{K}^n$ with ranked VMs in $\mathcal{R}^n$, and find the corresponding allocation prices. \\
\For{$s \gets 1 \; \mathbf{to} \; K$}
    {
    \If{$(1 \leq s < K)$}
    {
        $x_{s,r_{is}}^n = 1$ \\
        $p_s^n = \Theta_{R_s}^n b_{i,r_{s+1}}^n$
    } \ElseIf{$(s==R^n)$}{
        $p_s^n = b_{i,r_{s}}^n + \epsilon$} 
    }
\end{algorithm}
%
\subsection{Repeated GSP-Based Dynamic Resource Allocation and Pricing Mechanism}
%
Now we consider the dynamic computation offloading environment. We model the repeated GSP mechanism as a dynamic game of incomplete information, where the number of offloading tasks and the workloads of the VMs are uncertain until the offloading requests arrive at the processors. We summarize the repeated GSP-based resource allocation and pricing method in \textbf{Algorithm~\ref{algo:MEC_Repeated_GSP}}. The algorithm begins with all the state information on UEs and VMs as inputs. At regular intervals with length $T$, the algorithm updates the priority indices for each processor $n$ using (\ref{eqn:priority}) based on UEs' offloading historical data. \\ \indent
The algorithm continues the provisioning of offloading services at time $t=1, 2, \dots $ with a duration of $\Delta t$. At the beginning of each slot $t$, it gathers incoming offloading requests and places them into the queue $\mathcal{K}^n(t)$ according to the requested application type $n$. Any request that arrives after that will be processed in the next time slot. After queuing the offloading task in each processor $n \in \{1, 2, \dots, N\}$, it initiates the bid collection process by sending a request for bids to the servers along with information on tasks' priority indices $\boldsymbol{\lambda}^n$ and resource allocation decisions from the previous round, i.e. $\pi^n(t-1)$. For initialization, it uses random task assignments. \\ \indent
Next, the algorithm updates the expected computation performance quality scores, $\theta_{i,r_m}^n(t)$, for the VMs based on their current workloads. Upon receiving the bids from all the servers, each processor runs \textbf{Algorithm~\ref{algo:GSP_allocation_pricing}} to obtain the allocation and pricing decisions $\mathbf{\pi}^n(t)$ for the current round. Finally, the algorithm updates the workloads of the VMs according to the new allocation decisions before moving to the next auction round.
\begin{algorithm}[t]
\caption{$\mathit{Repeated\_GSP\_for\_MEC\_Offloading}$}
\label{algo:MEC_Repeated_GSP}
\KwInput{$\mathcal{I},\mathcal{J},\mathcal{N}\}$}
\KwOutput{$\left \{ \pi^n(t) \right \}$} \vspace{0.25em}
Initialize offloading service provisioning states,
$$\mathcal{K}^n \gets \emptyset, \mathcal{R}^n \gets \emptyset, \mathcal{X}^n \gets \emptyset, \mathcal{P}^n \gets \emptyset$$ \\
Update task priority scores at every interval $T$. \\
\For{$n=1 \; \mathbf{to} \; N$}{ \LinesNotNumbered
    \For{$j=1 \; \mathbf{to} \; K$}{
        $\lambda_{s_j}^n = \dfrac{\Hat{\gamma}_{i',j} \beta^n}{\Hat{d}_{s_j}^n \tau_{\max}^n f_{C_{\min}}^n}$
    }
}
Begin offloading service provisioning, \\ 
\While{$(t  \in \{1, 2, \dots \})$}{
    Collect the incoming offloading requests according to the requested task type, \\
    \For{$n=1 \; \mathbf{ to } \; N$}{
        \For{$s=1 \; \mathbf{to} \; K$}{
        $K^n(t)[s] = d_{k_j}^n (t)$
        }
        Request servers to submit their bids,
        $$\textbf{\textit{b}}^n(t) \gets \mathit{Request\_Bids \left(\lambda_{\mathrm{sorted}}^n(t), \pi^{n}(t-1)\right)}$$ \\
        Update VM's computation quality scores, \\
        \For{$i=1 \; \mathbf{to} \; I$}{
            \For{$m=1 \; \mathbf{to} \; R^n_i$}{
                Find resource utilization metric $\phi_{i,r_m}^n(t)$ using eqn. (\ref{eqn:VM_utilization}) \\
                $\theta_{i,r_m}^n(t) = \frac{W_{i}^n f_{C_i}^n}{C_{i}^n} \phi_{i,r_m}^{n}$
            }
        }
        Allocate resources and determine the prices,
        \begin{align*}
            \pi^n(t) \gets \mathit{GSP\_Mechanism} \left( \boldsymbol{\lambda}^n, \boldsymbol{\theta}^n(t), \mathbf{b}^n(t) \right)
        \end{align*} \\
        Update the workloads for allocated VMs, \\
        \For{$i=1 \, \mathbf{to} \, I$}{
            \For{$m=1 \; \mathbf{ to } \; R_i^n$}{
                \For{$s=1 \; \mathbf{to} \; K$}{
                    \If{$(x_{s,r_{im}}^n (t) == 1)$}{
                        $\eta_{i,r_m}^n(t) = d_{k_j}^n(t)$
                        }
                    }
                }
            }
        }
    }
\end{algorithm}
\subsection{A Toy Example}
%
Consider a MEC offloading scenario as illustrated in \textbf{Figure~\ref{fig:GSP_example_1}}, consisting of $J=4$ offloading UEs, and $I=2$ servers offering MEC services for $N=3$ different applications. The MEC orchestrator gathers the incoming offloading requests and forwards them to the task processing queues. Each processor has a task queue with $K=4$ distinct positions, where the task offloaded by UE $j$ is placed in position $s=j$. In \textbf{Fig.~\ref{fig:GSP_example_1}}, the pool of VM resources represented by $\textit{\textbf{R}}^1$, $\textit{\textbf{R}}^2$, and $\textit{\textbf{R}}^3$, respectively.   \\ indent
%
\begin{figure}[hb!]
    \centering
    \includegraphics[scale=0.9]{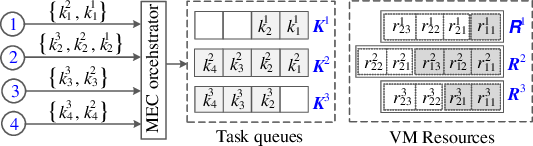}
    \caption{An illustration of incoming offloading requests and VM instances at MEC processors.}
    \label{fig:GSP_example_1}
\end{figure}
%
\begin{example}
\label{example_1}
Consider the MEC system, with the set of offloading tasks and VM resources as shown in \textbf{Fig.~\ref{fig:GSP_example_1}}. We present an example describing the proposed GSP mechanism. The task priority indices of UEs are given as $\boldsymbol{\lambda}^1 = [0.31, 0.20, 0.15, 0.09]$, $\boldsymbol{\lambda}^2 = [0.13, 0.23, 0.14, 0.38]$, and $\boldsymbol{\lambda}^3 = [0.26, 0.11, 0.24, 0.20]$. Besides, the VMs' quality scores and bids are given in \textbf{Fig.~\ref{fig:GSP_example_2}}.\\ \indent
At first, \textbf{Algorithm~\ref{algo:GSP_allocation_pricing}} sorts the arriving tasks in decreasing order of their priority scores in their respective task queues. As demonstrated in \textbf{Fig.~\ref{fig:GSP_example_2})}, the tasks in $\mathcal{K}^1$ are arranged as $\lambda_{s_1^1} > \lambda_{s_2^1}$. Similarly, the tasks in $\mathcal{K}^2$ and $\mathcal{K}^3$ are arranged according their priority scores. Next, the VMs in each processor are arranged according to their ranking scores. For example, VMs in $\mathcal{R}^1$ are ranked as $y_{1,r_1}^1>y_{2,r_2}^1>y_{2,r_3}^1>y_{2,r_1}^1$. Similarly, other VMs are ranked within $\mathcal{R}_2$ and $\mathcal{R}_3$.\\ \indent
Next, the algorithm sequentially matches the tasks in each queue $\mathcal{K}^n$ to the VMs in the queue $\mathcal{R}^n$ (as shown by directed arrows in \textbf{Fig.~\ref{fig:GSP_example_2}}). The allocation prices are determined based on the price adjustment rates and the bid of the VM ranked next. For example, task $k_1^1$ in the first position in $\mathcal{K}^1$ is matched to VM $r_{11}^1$, which is in the top rank in $\mathcal{R}^1$. The corresponding task assignment decision variable is then updated as $x_{1,r_{11}}^1 = 1$, and the allocation price yields $p_1^1 = \Theta_{R_1}^1 b_{2,r_2}^1 = (1.238 \times 0.22) = \$0.285$/VM-hr. The same procedure follows for other assigned tasks.
\end{example}
\begin{figure}[t]
    \centering
    \includegraphics[scale=0.9]{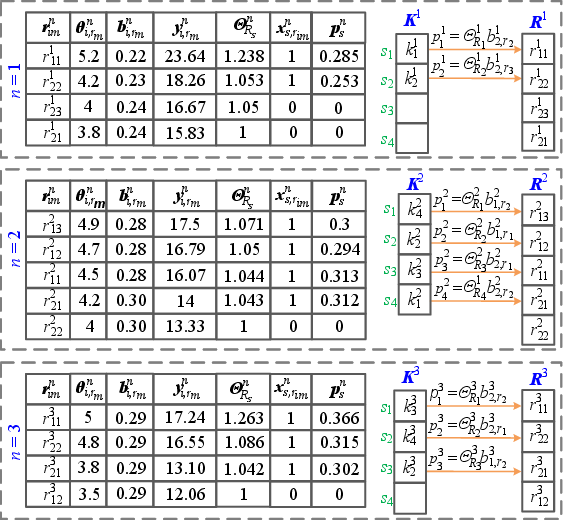}
    \caption{An example illustrating the GSP-based task assignment, VM allocation, and pricing.}
    \label{fig:GSP_example_2}
\end{figure}
\section{Analysis of Bidding Strategies in GSP-Based Dynamic MEC Offloading Auction}
\label{sec:bidding_analysis}
%
In this section, we study the strategic behavior of MEC servers. The servers who participate as bidders can adjust their bids in every time slot by behaving strategically based on the history. The servers thus learn about competitors' bids through interactions over successive auction rounds and adjust their strategies to maximize their utility. To that end, we study best-response bidding strategies for servers to adjust bids within a certain range. That also guarantees to achieve a Nash equilibrium in every round, which maximizes the utility for every participating server.
%
\subsection{Adaptive Balanced Bidding Strategies of Servers}
%
We consider myopic best-response strategies \cite{M_Cary_2014}, where every server $i$ adjusts its bids for the current auction round $t$ under the assumption that other servers repeat their bids of the previous round. In the dynamic computation offloading setting, the CPU utilization and the computation performance qualities of VMs fluctuate based on their computational workloads. Hence, we consider an adaptive balanced bidding policy, where each server devises the current bidding strategies considering the VMs' expected computation performance quality scores and adjusts the bids accordingly. \\ \indent 
In the position auction framework, the offloading tasks in the higher positions return higher profits for VMs, due to the monotonicity in GSP prices. Hence, every bidder adjusts its bids for each VM by targeting a higher task position that can maximize the VM's profit in the next auction round. To win a task with a higher offloading rate, a lower value is necessary so that the VM can obtain a higher-ranking position in the auction. On the contrary, self-interested bidders tend to bid higher for their VMs with higher expected quality scores to earn more profits. \\
Therefore, we propose a balanced bidding strategy that allows a server to adjust the bids for its VMs within a certain range, and yet maximize the profit. To restrict servers from overbidding, we use the restricted balanced bidding (RBB) strategy \cite{Cary_2007}: A server determines the bid for each of its VMs by aiming for the desired task position $s_m^*$ that maximizes the $m$-th VM's utility. It then adjusts its bid based on the allocation price of the target slot $p_{s_m^*}^n (t-1)$ in the previous round. The core concept of balanced bidding is limiting the degree to which a server acts greedy while adjusting the bids for VMs. In the RBB strategy, when a server adapts its bid for each VM $r_{im}^n$, it can only aim for task positions with no higher priority indices than the position that the $m$-th VM has already won in the previous round. That allows the server to bid as low as essential to secure the $s_m^*$-th rank in the next auction round. Moreover, it adjusts the bid in a balanced way so that if the VM cannot win the target slot, it does not end up with a profit lower than the previous round. Below, we define the RBB strategy formally.
\begin{definition}
\label{def:RBB}
Let $\Hat{s}_{m}$ be the position that the $m$-th VM at server $i$ has won in round $t-1$. Assuming that other servers' bids remain fixed at their previous values, the \textbf{Restricted Balanced Bidding (RBB)} strategy is to: 
\begin{itemize}
    \item[(i)] Find the target task position $s_m^*$ among the positions ranging from $\Hat{s}_m$ to $K$ that maximizes the VM's utility
    \begin{align}
    \nonumber s_m^* = \arg \max_{s'} \left \{ \lambda_{s'_{j}}^n \theta_{i,r_m}^n (t) \left(p_{s'}^n(t-1) - v_{i,r_{s'}}^n \right)  \right \}.
    \end{align}
    If the VM has not been allocated to any position in the previous auction round, then the server looks for the target position within the range $1 \leq s' \leq K$.
    \item[(ii)] Adjust the bid for the current auction round $t$ to satisfy
    \begin{align}
    \begin{split}
       \nonumber \lambda^{n}_{s_m^*} \theta_{i,r_m}^n(t-1) \left( p_{s_m^*}^n(t-1) - v_{i,r_m}^n \right) \\ = \lambda^{n}_{s_m^*-1} \theta_{i,r_m}^n(t) \left( b_{i,r_m}^{n}(t) - v_{i,r_m}^n \right). 
    \end{split}
    \end{align}
\end{itemize}
\end{definition}
%
Using the RBB strategy as defined above, the updated bid for the current auction round $t$ for the $m$-th VM at server $i$ yields
\begin{equation}
\label{eqn:RBB}
    b_{i,r_m}^{n}(t) = v_{i,r_m}^n + \Pi_{i,r_{m}}^n \left( p_{s_m^*}^n (t-1) - v_{i,r_m}^n\right),
\end{equation}
where $\Pi_{i,r_{m}}^n = \dfrac{\lambda^{n}_{s^*_m} \theta_{i,r_m}^n(t-1)}{\lambda^{n}_{s^*_m-1} \theta_{i,r_m}^n(t)}$. If the target task position is the topmost one, i.e. $s_m^*=1$, then we assume $\lambda^{n}_{0}  = 2 \lambda^{n}_{1}$ to adjust the bid. In \textbf{Algorithm~\ref{algo:RBB}}, we describe adapting the bid by a server following the RBB bidding strategy for the VMs in every processor $n$.
%
\begin{algorithm}[ht]
\caption{$\mathit{Restricted\_Balanced\_Bidding\_Strategy}$}
\label{algo:RBB}
\KwInput{$\mathbf{\lambda}^{n}, \,  \textbf{\textit{x}}^n (t-1), \,  \textbf{\textit{p}}^n(t-1)$} \vspace{1ex}
\KwOutput{$\textit{\textbf{b}}_i^n(t)$}
Upon receiving bid request from processor $n$,\\
\For{$m=1 \; \mathbf{to} \; R_i^n$}
{
    Find whether the VM has won any task position in the previous time slot $(t-1)$, \\
    \For{$s=1 \; \mathbf{to} \; K$}
    {   
        \If{$(x_{s,r_{im}}^n (t-1)== 1)$}{
            $\Hat{s}_{m} = s$ \\
            \Else{
            $\Hat{s}_{m} = 1$
            }
        }
    }
    Find the target slot $s_m^*$ that maximizes utility, \\  
    \For{$s'=\Hat{s}_{m} \; \mathbf{to} \; K$}{
    $\Hat{u}_{s',m}^n = \lambda_{s'_{j}}^n \theta_{i,r_m}^n(t) \left( p_{s'}^n(t-1) - v_{i,r_{s'}}^n \right)$
    }
    $$s_m^{*} \gets  \arg \max(\Hat{\textbf{\textit{u}}}^n)$$\\
    Adjust bid according to RBB strategy in eqn.~(\ref{eqn:RBB}), \\
    \If{$(s_m^*==1)$}{$\Pi_{i,r_m}^n = \dfrac{ \theta_{i,r_m}^n(t-1)}{2 \, \theta_{i,r_m}^n(t)}$}
    \Else{$\Pi_{i,r_m}^n = \dfrac{\lambda^{n}_{s^*_m} \theta_{i,r_m}^n(t-1)}{\lambda^{n}_{s^*_m-1} \theta_{i,r_m}^n (t)}$}
    $b_{i,r_m}^{n}(t) \gets v_{i,r_m}^n + \Pi_{i,r_m}^n \left( p_{s_m^*}^n (t-1) - v_{i,r_m}^n \right)$ 
}
\end{algorithm}
\subsection{Analysis of Bidding Dynamics on Auction Efficiency}
%
In this section, we analyze the efficiency of the GSP mechanism under dynamic MEC offloading setting considering synchronous bidding model \cite{Varian_2007}, where servers update their bids simultaneously. We show that no server encounters a negative utility gain through their participation in the auction satisfying the \textit{Individual Rationality (IR)} property, when each server follows the RBB strategy in every round of auction (\textbf{Theorem~\ref{theo:IR}}). First, we define the Individual Rationality (IR) property.
%
\begin{definition}
\label{def:IR}
In the GSP-based MEC offloading mechanism, the \textbf{Individual Rationality (IR)} is satisfied, if the resource allocation prices $p_s^n(t)$ guarantee non-negative utility gain, i.e., $u_{i,r_m}^n(t) \geq 0$ for every VM $r_{im}^n$ at server $i$ that participates in the computation procedure at processor $n$ during the time slot $t$.
\end{definition}
\begin{theorem}
\label{theo:IR}
The GSP-based MEC offloading auction at each processor $n$ guarantees the individual rationality for every participating VM, when the servers follow the proposed RBB strategy in every round of auction.
\end{theorem}
\begin{proof}
See \textbf{Appendix~A}.%
\end{proof}
%
Next, we analyze the stability of the GSP-based resource allocation outcomes in a dynamic setting, the same auction mechanism is repeated in every time slot $t$ with new sets of offloading requests at $N$ processors. In this scenario, the concept of stability in the auction mechanism in each processor $n$ is represented in terms of an equilibrium point with a set allocation price decisions, where every server is well-off with the resource allocation decisions and do not wish to exchange any of the VM allocation with another VM during the time slot $t$. Therefore, the auction reaches $N$ distinct equilibrium points at different processors in every time slot $t$. In Theorem~\ref{theo:SNE}, we show that the proposed GSP-based MEC offloading mechanism results into a set of \textit{Symmetric Nash Equilibrium (SNE)} allocation prices in every auction round, where no server prefers to exchange the assigned task positions for any of its VM, with another task position within the same processor. Instead the servers are well-off with their utilities at SNE, and thus maintain the equilibrium by following the RBB strategy for the future auction round. We formally define the SNE allocation prices in a dynamic MEC offloading scenario, as follows:
%
\begin{definition}
\label{def:SNE}
The set of resource allocation prices $\textbf{\textit{p}}^n(t)$ at processor $n$ during the time slot $t$, is in \textbf{Symmetric Nash Equilibrium (SNE)} if the following holds for any task positions $s$ and $s'$ in $\mathcal{K}^n(t)$:
\begin{align}
\label{eqn:SNE_ineq}
    \hspace{-1em} \lambda_{s_j}^n \left(p_s^n(t) - v_{i,r_s}^n \right) \geq \lambda_{s'_j}^n \left(p_{s'}^n(t) - v_{i,r_{s'}}^n \right).
\end{align}
\end{definition}
\begin{theorem}
\label{theo:SNE}
There exists a set of symmetric Nash equilibrium (SNE) allocation prices for GSP-based MEC offloading auction at each processor $n$, when every server $i$ follows the proposed RBB strategy in every round of auction.
\end{theorem}
\begin{proof}
 See \textbf{Appendix B}.
\end{proof}
%
In the dynamic setting, the auction has a set of SNEs and thus may reach different SNE points at various auction rounds. However, the bidders can vary their bids within a certain range without violating the stability of the allocation outcomes. Theorem~\ref{theo:SNE_UP_LB} presents the upper and lower bound to the SNE for the proposed GSP-based resource allocation mechanism.
\begin{theorem}
\label{theo:SNE_UP_LB}
For any task type $n$, the upper- and lower bounds for the bid to satisfy the SNE conditions are given by
\begin{align}
\label{eqn:UB_SNE}
       \nonumber b^{n,\mathrm{UB}}_{i,r_s} x^n_{(s-1),r_{(s-1)}} & = \frac{v^n_{i,r_s}}{\Theta_{R_{(s-1)}}^n} x^n_{(s-1),r_{(s-1)}} \\ &  + \frac{\lambda^*\Theta_{R_{s}}^n}{\Theta_{R_{(s-1)}}^n} \left( b^n_{i,r_{(s+1)}} - v^n_{i,r_s} \right) x^n_{s,r_s},
\end{align}
\begin{align}
\label{eqn:LB_SNE}
       \nonumber b^{n,\mathrm{LB}}_{i,r_s}  x^n_{(s-1),r_{(s-1)}} & = \frac{v^n_{i,r_{(s-1)}}}{\Theta_{R_{(s-1)}}^n} x^n_{(s-1),r_{(s-1)}}  \\ & + \frac{\lambda^*\Theta_{R_{s}}^n}{\Theta_{R_{(s-1)}}^n} \left( b_{i,r_{(s+1)}}^n - v_{i,r_{(s-1)}}^n \right) x_{s,r_s},
\end{align}
where $\lambda* = \frac{\lambda_{s_j}^n}{\lambda_{(s-1)_j}^n}$, $r_{(s-1)}$, $r_s$, and $r_{(s+1)}$ denotes the VMs allocated to the task positions $(s-1)$, $s$, and $(s+1)$, respectively. 
\end{theorem}
\begin{proof}
See \textbf{Appendix C}.
\end{proof}
%
To validate that the proposed GSP-based MEC offloading auction reaches the equilibrium point within a polynomial time in every auction round, we analyze the computational complexity of the proposed resource allocation and pricing algorithms. Using \textbf{Algorithm~\ref{algo:MEC_Repeated_GSP}} in every time slot $t$, one can obtain the computation offloading task assignment and allocate VM resources. The algorithm has prior knowledge about the task priority indices. Theorem~\ref{theo:CE} states some results about the computational efficiency of our proposed mechanism.
%
\begin{theorem}
\label{theo:CE}
The proposed GSP-based MEC offloading auction is computationally efficient, i.e. the resource allocation outcomes are determined in polynomial time in every auction round $t$.
\end{theorem}
\begin{proof}
See \textbf{Appendix D}.
\end{proof}
\section{Numerical Results} 
\label{sec:results}

We present the numerical results on the convergence properties of the proposed RBB bidding strategy, and the performance of the proposed GSP auction mechanism. We consider a simulation setup for the MEC system that manages computation offloading service provisioning within a $ 250 \times 250 \mathrm{m}^2 $ area. There are $I=5$ servers/nodes and $J=150$ wireless UEs, randomly located within this area following uniform and non-uniform distribution, respectively. The UEs are associated with the nearest physical MEC node. The offloading requests are generated randomly in each auction round, where the offloading data sizes follow the Poisson distribution with the parameter $d_{\mathrm{avg}}$. \textbf{Table~\ref{tab:parameters}} lists the MEC offloading and wireless channel parameters. \textbf{Table~\ref{tab:VM_model}} provides the VM configuration model. Each server offers a single type of VM instance compatible with the $n=1$-th MEC application. To capture the dynamics of the wireless channel and offloading process more accurately, we simulate $1000$ times for each time slot and use the average for that auction round.
\begin{table}[t!]
\caption{MEC offloading and wireless channel parameters}
\label{tab:parameters}
\begin{tabular}{l|l}
    \toprule
     $N=1, \, I=2$  &   $J = 150, \, K=J \, \Delta_t = 1$ min \\
    $f_{C_{\min}} = 3.2$ GHz& $F_t$= 5.8 GHz \\
    $[\tau_{\min}, \tau_{\max}] = [20,200]$ msec & $P_u$ = 20 dBm \\
    $\mathrm{BW}$ = 80 MHz  & $\sigma_N^2$ = -100 dBm \\
    $[\Gamma_{\min},\Gamma_{\max}] = [0.01, 0.90]$ & $\mu_d=2.12, \mu_0 = 29.2, \mu_{f} = 2.11$ \\
    $d_k \sim \mathrm{Poisson}(d_{\mathrm{avg}})$ MB &  $d_{\mathrm{avg}} \sim  [10,40]$ MB\\
    $\kappa = 10^{-24}, \epsilon = 0.001$ &    $q_c = 0.5, q_l = 0.5$   \\
    $\Bar{a} = \$\, 20$/month  &  \\
    \bottomrule
\end{tabular}
\end{table}
\begin{table}[t]
\setlength{\belowcaptionskip}{-5pt}\caption{VM configuration model}
\label{tab:VM_model}
\begin{tabular}{c|c|c|c|c|c|c}
    \toprule
    \multirow{2}{*}{MEC Server} & \multirow{2}{*}{$R_i^n$}  & \multirow{2}{*}{$W_i$} & \multirow{2}{*}{RAM} & \multirow{2}{*}{$f_{C_i}^n$} & \multirow{2}{*}{$C_i$}  & \multirow{2}{*}{$\rho_i$} \\ 
    \multirow{2}{*}{$i$} &  &  & \multirow{2}{*}{(GB)} & \multirow{2}{*}{(GHz)} & \multirow{2}{*}{(MIPS)} & \multirow{2}{*}{($\$$/VM-hr)} \\ [2ex]
    \hline
     $i=1$ & 60 &  2 & 4 & 3.3 GHz & 32  &  0.0452 \\ 
     $i=2$ & 60 &  2 & 4 & 3.5 GHz & 24  &  0.0435 \\ 
     $i=3$ & 40 &  2 & 4 & 3.2 GHz & 24  &  0.0385 \\ 
     $i=4$ & 50 &  1 & 2 & 3.2 GHz & 16  &  0.0186 \\ 
     $i=5$ & 50 &  1 & 2 & 3.3 GHz & 16  &  0.0175 \\ 
    \bottomrule
\end{tabular}
\end{table}

\subsection{Convergence of Bidding Strategies}
To analyze the convergence properties for the proposed GSP-based mechanism, we consider a static offloading scenario assuming all $J=150$ UEs offload the same length of computing tasks in every round of auction. Furthermore, there are $I=2$ servers with private valuations $v_1 = \$ \, 0.0354$ and $v_2 = \$ \, 0.0366$, respectively. We study their competitive bidding behavior by following the proposed RBB strategy. We investigate three cases with the various number of available VMs, i.e., $R = [150, 1]$, $R=[1, 150]$, and $R=[80, 80]$, representing different levels of competition among the servers. The first two cases represent the monopoly market scenario where server $i=1$ and server $i=2$ dominate the supply of VM resources, respectively. The third case corresponds to fair competition among servers. To identify the convergence point, we set the convergence tolerance to $0.0001$.\\ \indent
\textbf{Fig.~\ref{fig:bid_conv_1}} shows the average bids of the servers and the respective allocation prices for the above three cases. The solid- and dashed lines represent the bids submitted by the server $i=1$ and the server $i=2$,  respectively. In all three cases, the servers bid no less than their private valuations (as shown by the green lines in \textbf{Fig.~\ref{fig:bid_conv_1}(a)}), ensuring non-negative profit/utility gain. Every server tends to increase its bid prices to beat the opponent, which gradually converges to some fixed points after a finite number of auction rounds. \textbf{Fig.~\ref{fig:bid_conv_1}(b)} shows the corresponding allocation prices, where we compare the proposed GSP-based allocation prices with the VCG-based prices obtained from the Knapsack-based dynamic programming solution approach for (\ref{eqn:opt_prob}). In all three cases, the VCG-based allocation prices are constant and lower than the proposed GSP-based prices, confirming that VCG mechanism is welfare-maximizing (i.e., buyer-friendly). \\ 
%
\begin{figure}[t!]
    \centering
    \includegraphics[scale=0.6]{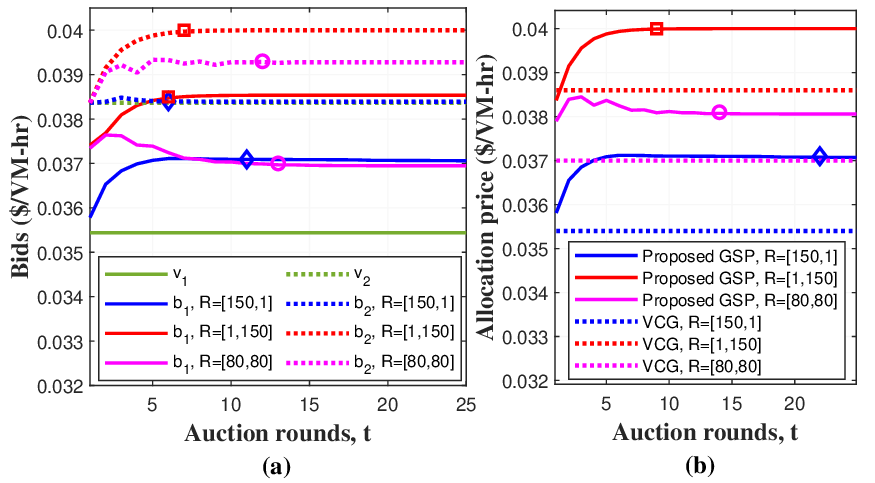}
    \setlength{\belowcaptionskip}{-10pt}\caption{Convergence analysis for the proposed RBB bidding strategies, based on (a) average bid prices of each server, $b_i^n$, and (b) average allocation prices, $p^n$, with $I=2$ servers and varying number of VMs, $R_i^n$.}
    \label{fig:bid_conv_1}
\end{figure}
%

In the first case, the bids submitted by server $i=1$ and server $i=2$ converge to $b_1^* = \$ \, 0.0371$/VM-hr and $b_2^* = \$ \, 0.0384$/VM-hr, respectively (\textbf{Fig.~\ref{fig:bid_conv_1}(a)}). In this case, server $i=1$ dominates the market supply. Thus the majority of the offloading tasks go to VMs at that server. The corresponding allocation prices reach the equilibrium point $p^* = \$ \, 0.0371$/VM-hr (\textbf{Fig.~\ref{fig:bid_conv_1}(b)}). A similar trend is observed for the second case, where server $i=2$ dominates the market supply, and wins the majority of the tasks. In this case, server $i=1$ tries to beat the opponent despite its lower supply of VM resources, which is reflected in his/her increased bids. When the market has a fair level of competition as in the third case, the average bids of both servers are decreased. The corresponding allocation prices, shown in \textbf{Fig.~\ref{fig:bid_conv_1}(b)}, confirm that bringing more sellers into the offloading service market influences the market to settle at a lower equilibrium price, $p^*(t=14) = \$ 0.0381$/VM-hr, which is more favorable to the UEs.\\ \indent 
%
\begin{figure}[b!]
    \centering
    \vspace{-2 em}
    \includegraphics[scale=0.62]{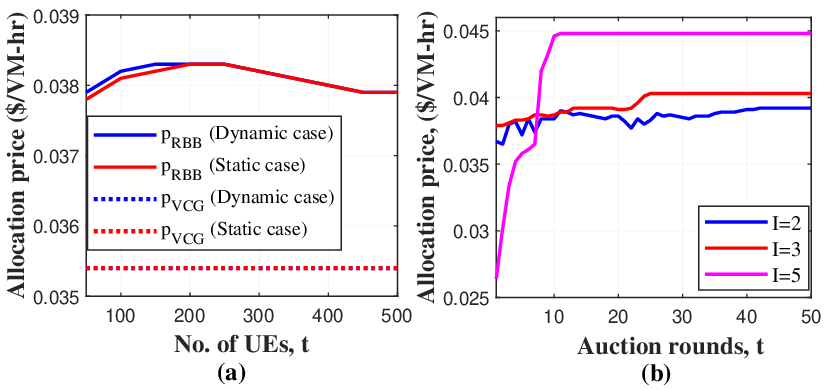}
    \caption{Comparison of average allocation prices for (a) varying number of UEs, $J$; and (b) MEC servers, $I$.}
    \label{fig:bid_conv_2}
\end{figure}
%
Next, we analyze the convergence of the proposed RBB strategies considering the dynamic offloading scenario, where the size of the computing tasks changes following the Poisson distribution. In \textbf{Fig.~\ref{fig:bid_conv_2}(a)}, we compare the average allocation prices for the proposed GSP mechanism with the VCG mechanism for varying numbers of UEs, starting from $J=50$ to $J=500$. Furthermore, we assume the size of the offloading task queue is the same as the number of UEs, i.e., $K=J$, and the number of VMs at the servers as $\textbf{\textit{R}}=[250, 250]$. As shown in \textbf{Fig.~\ref{fig:bid_conv_2}(a)}, the average allocation prices for both the static and dynamic cases are very close for a lower number of UEs. Eventually, as the number of UEs increases, they coincide. In the beginning, with $J=50$, the resource utilization is low. Hence, the proposed GSP mechanism selects low allocation prices, i.e., $p_{\mathrm{GSP}} = \$ \, 0.0378$/VM-hr and $p_{\mathrm{GSP}} = \$ \, 0.0379$/VM-hr for the static and dynamic cases, respectively. As the number of UEs, thus resource utilization and competition among the server, increases, the corresponding allocation prices increase. They reach the maximum of $p_{\mathrm{GSP}} = \$ \, 0.0383$/VM-hr for both the static and dynamic cases at $J=250$. At this point, the servers exhibit the maximum level of competition since, for each UE, the probability of being matched with either server is equal. When the number of UEs grows to even larger values, the allocation prices start to drop because the increasing workloads at the VMs restrict the servers from raising the bids further. In both the static and dynamic cases, the VCG-based allocation prices remain fixed at $p_{\mathrm{VCG}} = \$ \, 0.0354$/VM-hr, lower than the proposed GSP-based prices. \\ \indent
In \textbf{Fig.~\ref{fig:bid_conv_2}(b)}, we compare average allocation prices for the proposed GSP mechanism for different numbers of resource sellers in the offloading service market. Considering a dynamic offloading scenario with $J=120$ UEs and a total of $R=120$ VMs, we study three different cases with $I=2$, $I=3$, and $I=5$, each server having $R_i = 60$, $R_i = 40$, and $R_i = 24$ VMs, in each case respectively. As the number of servers increases, the competition among the resource sellers grows, which results increased allocation prices. When the number of servers $I=5$, the average allocation prices converge to $p^* = \$ \, 0.0448$/VM-hr which is approximately $11 \% $, and $14 \%$ higher than the cases with $I=2$ and $I=3$, respectively.
%
\subsection{Auction Revenue and Profits of Servers}
To evaluate the auction performance of the proposed GSP-based offloading mechanism, we compare the results with the other greedy bidding strategies in GSP mechanism \cite{M_Cary_2007}, such as: (i) balanced bidding (BB) strategy, where servers can target any task position maximizes its utility gain, (ii) altruistic bidding (AB) strategy, where servers always bid lower than the allocation price of the target task position favoring the buyers, and (iii) competitor busting (CB) strategy, where the servers act vindictively by bidding higher than the allocation price of the target task position so that the competitor ends up with a lower utility gain. \\ \indent
%
\begin{figure}[t!]
    \includegraphics[scale=0.6]{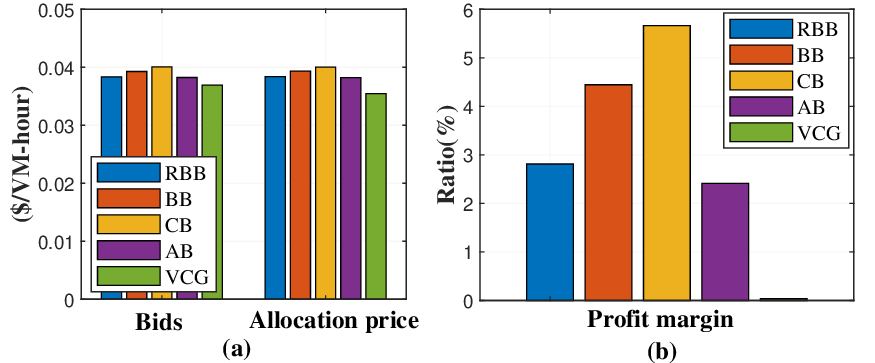}
    \caption{Performance comparison of proposed RBB and other existing bidding strategies, based on: (a) Submitted bids vs allocation prices; (b) Profit margin ratio (\%) of servers.}
    \label{fig:MEC_profit_1}
\end{figure}
Considering $I=2$ servers with $R=[80,80]$ VMs and $J=150$ UEs, we first compare the equilibrium bids (i.e., the point of bid convergence) and allocation prices of the proposed strategy with other greedy strategies and truthful bidding in VCG mechanism. As shown in \textbf{Fig.~\ref{fig:MEC_profit_1}(a)}, the CB strategy results in the highest average bid, thus the highest allocation price; nevertheless, it is not suitable for both sellers and buyers. Among other strategies, the BB strategy gives higher bids and allocation prices than the proposed RBB strategy. That is because the BB strategy allows servers to choose their preferable target positions without any restriction, unlike our proposed strategy. The BB strategy also performs well in terms of a profit margin ratio of $4.44\%$ (\textbf{Fig.~\ref{fig:MEC_profit_1}(b)}). However, our proposed RBB strategy results in $2.811\%$, which is still better than the AB strategy and the VCG mechanism. \\ \indent
\begin{figure}[t!]
    \centering
    \includegraphics[scale=0.6]{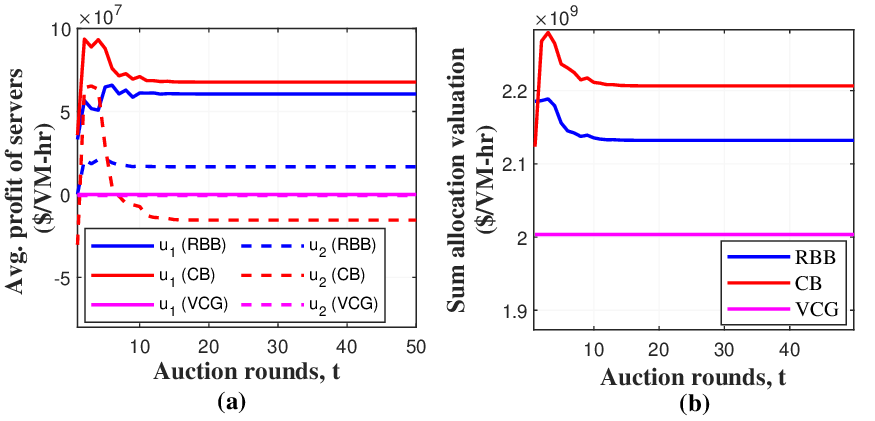}
    \setlength{\belowcaptionskip}{-15pt}\caption{Performance evaluation of proposed RBB and other existing bidding strategies, based on (a) average profits of servers; and (b) sum allocation valuation }
    \label{fig:MEC_profit_2}
\end{figure}
Next, we analyze the average profit of each server and the auction revenue, i.e., the sum valuation of task computation at the allocated resources. We compare the results for the proposed RBB strategy with the truthful bidding in the VCG mechanism and the CB strategy, representing the best-case and the worst-case scenarios from the users' perspectives. As shown in \textbf{Fig.~\ref{fig:MEC_profit_2}(a)}, the average profits of server $i=1$ and $i=2$ for the proposed RBB strategy converge to $u_i^* = \$ \, 6.07 \times 10^7$/VM-hr and $u_i^* = \$ \, 1.67 \times 10^7$/VM-hr, respectively. For the CB strategy, although the average profit of server $i=1$ becomes even higher than the proposed RBB strategy (shown with the solid red line), server $i=2$ ends up with a negative utility gain because of the vindictive bidding behavior. The servers' gains for the VCG mechanism remain zero. A similar trend appears in \textbf{Fig.~\ref{fig:MEC_profit_2}(b)}, where the proposed RBB strategy yields a higher sum allocation valuation than the VCG mechanism, whereas the CB strategy exceeds the total valuation due to higher allocation prices. 
\subsection{Social Welfare and QoE Analysis}

In this section, we analyze the performance of the proposed GSP-based offloading mechanism addressing the QoE of UEs and the overall social welfare of the MEC system. With $I=3$ servers and $K=300$, we compare the social welfare and the average utility of UEs of the proposed GSP mechanism with the VCG mechanism in \textbf{Fig.~\ref{fig:SW_UE_utility}}. To that end, we change the number of UEs from $J=50$ to $J=550$. We evaluate the results for two cases, $R=300$ and $R=450$, with $R_i=100$ VMs and $R_i=150$ VMs in each server $i$, respectively. \\ \indent 
%
\begin{figure}[ht!]
    \hspace{-0.8em}
    \includegraphics[scale=0.6]{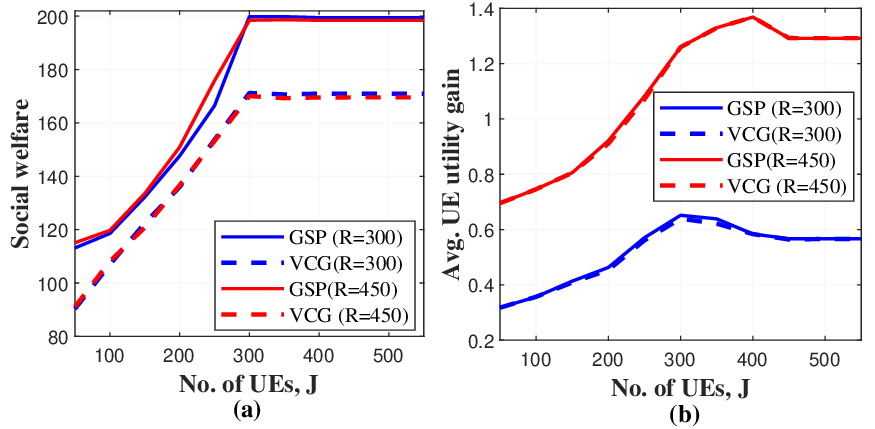}
    \setlength{\belowcaptionskip}{-10pt}\caption{Performance comparison between the proposed GSP and VCG mechanism for varying number of VMs, based on (a) Social welfare; and (b) Average utility gain of UEs.}
    \label{fig:SW_UE_utility}
\end{figure}
%
\textbf{Fig.~\ref{fig:SW_UE_utility}(a)} shows that the social welfare for the GSP- and VCG mechanisms increases as the number of UEs increases. It then remains fixed after the cut-off point $J=K=300$ when the number of UEs surpasses the capacity limit. Furthermore, the proposed GSP mechanism provides higher social welfare than the VCG mechanism for both cases. Although the VCG mechanism guarantees higher utility gain for the UEs than the proposed GSP mechanism (\textbf{Fig.~\ref{fig:SW_UE_utility}(b)}), the latter results in higher social welfare than the former because the profit of the servers is always higher. \textbf{Fig.~\ref{fig:SW_UE_utility}(b)} also reveals that the average utility gain of UEs for the GSP- and VCG mechanisms coincide for both cases. When the MEC system has the same resource availability as the system capacity limit, i.e., $R=300$, then the average utility gain increases until the number of UEs hits the system capacity, i.e., $K=300$. It further decreases until it reaches the fixed point at $Q^* = 0.567$. For the second case with $R=450$, the average utility gain of UEs improves with more supply of VM resources for both mechanisms. In this case, the UEs' average QoE reaches its maximum when $J=400$ and then goes down to the fixed point of $Q^* = 1.291$. In the case of social welfare, the changes in VM resource supply (e.g., $R=300$ and $R=450$) have no significant impact on the mechanisms (\textbf{Fig.~\ref{fig:SW_UE_utility}(b)}). \\ \indent 
Next, we study the QoE of UEs for three different cases by varying the average offloading task sizes: (i) $d_\mathrm{avg} = [5,20]$, (ii) $d_\mathrm{avg} = [10, 40]$, and (iii) $d_\mathrm{avg} = [20, 100]$. With $R=450$ and $K=400$, we plot the average offloading cost and service latency of UEs in \textbf{Fig.~\ref{fig:QoE_cost_latency}}. The results are almost the same for both mechanisms. The average offloading cost and service latency of UEs tend to increase as the number of UEs increases until it reaches the system's capacity limit, i.e., $K=400$, and then remain fixed. Moreover, when the average offloading data sizes increase, the corresponding average offloading cost and service latency of UEs rise significantly. To be specific, when the average offloading task size is $d_\mathrm{avg} = [20, 100]$, the average offloading service latency exceeds the maximum task completion threshold, $\tau_{\mathrm{max}} = 200$ msec, after the number of UEs reaches $J=200$. The average offloading cost for the case with $d_\mathrm{avg} = [20, 100]$ is more than double the offloading compared to others.

\section{Conclusion}
\label{sec:conclusion}
We have presented a dynamic resource allocation mechanism based on the SDN-enabled MEC system architecture, which implements computation offloading as a service via an auction process. The main objective of our proposed auction-based computation offloading mechanism is to maximize the sum valuation of the computing resources at the servers while satisfying the QoE of offloading users. In a dynamic setting, where the number of offloading requests, the corresponding task lengths, and the computational workloads of the servers vary over time, we formulate the offloading task assignment and resource allocation decision problem. As this becomes an instance of the NP-hard online multiple choice multiple Knapsack problem, we have developed a modified repeated GSP auction-based resource allocation and pricing algorithm, which is computationally efficient and satisfies individual rationality property. Furthermore, we have proposed a restricted balanced bidding (RBB) strategy for the servers, which guarantees the symmetric Nash equilibrium (SNE) condition for the proposed GSP mechanism in every round of auction. Using the VCG mechanism in combination with the Knapsack dynamic programming-based resource allocation as the benchmark solution, we have provided extensive numerical analysis to evaluate the convergence properties of the proposed RBB strategy and the performance of the proposed GSP mechanism from the perspectives of both the servers and UEs.

\begin{figure}[t!]
    \hspace{-0.8em}
    \includegraphics[scale=0.6]{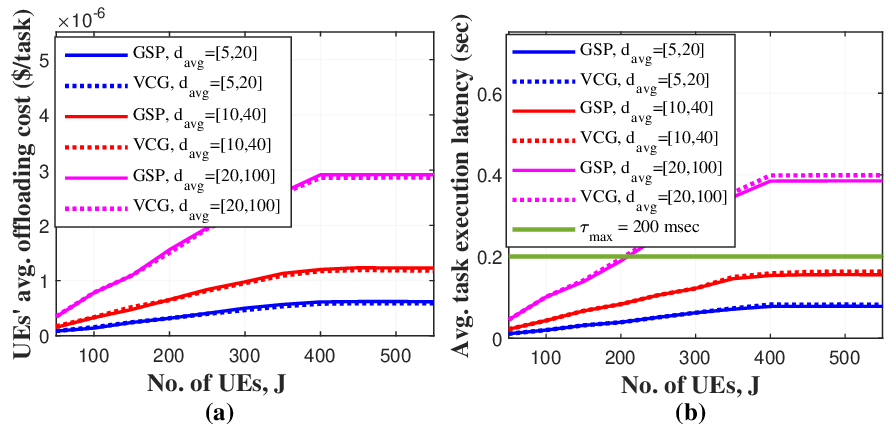}
    \setlength{\belowcaptionskip}{-9pt}\caption{QoE Analysis for the proposed GSP- and VCG mechanisms with varying offloading task sizes, showing (a) average offloading cost; and (b) service latency.}
    \label{fig:QoE_cost_latency}
\end{figure}

\bibliographystyle{IEEEtran}
\bibliography{References.bib}

\end{document}